\documentclass[11pt]{article} 
\usepackage{amsmath,amssymb,amsfonts,amsthm}

\usepackage[margin=1in]{geometry}
\usepackage{aliascnt}
\usepackage{hyperref}
\usepackage{enumitem}
\usepackage{float} 
\usepackage{placeins}
\usepackage{mathtools} 
\usepackage{todonotes} 
\usepackage{tikz-cd}
\newcommand{\R}{\mathbb{R}}
\newcommand{\C}{\mathbb{C}}
 
\usepackage{graphicx} 
\theoremstyle{definition}
\newtheorem{theorem}{Theorem}[section] 
\newtheorem{lemma}[theorem]{Lemma} 
\newtheorem{definition}[theorem]{Definition}
\newtheorem{proposition}[theorem]{Proposition}
\newtheorem{corollary}[theorem]{Corollary} 
\newaliascnt{remark}{theorem}
\newtheorem*{theorem*}{Theorem}
\newtheorem{remark}[remark]{Remark}
\aliascntresetthe{remark}

\DeclareMathOperator{\Spec}{Spec}
\DeclareMathOperator{\ind}{ind}

\DeclareMathOperator{\sign}{sign}
 
\DeclareMathOperator{\APS}{APS}

\DeclareMathOperator{\Dom}{Dom} \newcommand{\Z}{\mathbb{Z}}  
\hypersetup{colorlinks,breaklinks, linkcolor = purple, citecolor = teal, urlcolor = purple, }
\numberwithin{equation}{section}

\newcommand{\Id}{\text{Id}}

\DeclareMathOperator{\Tr}{Tr}

\title{Reflection Symmetry, APS Boundary Conditions, and Equivariant Spectral Flow on a Warped Cylinder}
\author{Taro Kimura$^*$, Sanchita Sharma%
\thanks{Université Bourgogne Europe, CNRS, IMB UMR 5584, 21000 Dijon, France}}
\date{}

\begin{document}
\maketitle


\begin{abstract}
We study reflection symmetry and Atiyah-Patodi-Singer (APS) boundary conditions for twisted Dirac operators on a finite warped cylinder. For a complex line twist with holonomy parameter $A$, we show that the reflection 
lifts to a unitary symmetry of the twisted Dirac setting if and only if $2A\in\mathbb Z$. In the resulting reflection-compatible fixed-holonomy case, reflection pairs opposite shifted angular modes, and the paired APS blocks are unitarily equivalent. 
The reflection trace on the APS harmonic space localizes to the unique self-paired zero-mode sector. We then turn to parameter-dependent versions of the model. For fixed gauge-trivial holonomy, the family remains pointwise \(O(2)\)-equivariant, and its spectral flow admits an \(RO(O(2))\)-valued decomposition. For genuinely varying holonomy, pointwise \(O(2)\)-equivariance is lost along the path. The representation-ring-valued invariant is then replaced by a residual sign-level invariant: the mod-two parity of the APS crossing events.

\end{abstract}

\setcounter{tocdepth}{2}
\tableofcontents


\section{Introduction}

The Atiyah-Patodi-Singer (APS) boundary conditions are among the basic nonlocal elliptic boundary conditions for Dirac-type operators in manifolds with boundary \cite{APS1,APS2,APS3}. Defined from the spectral decomposition of the induced self-adjoint boundary Dirac operator, they play a central role in index theory, spectral asymmetry, and spectral flow; see, for example, \cite{LawsonMichelsohn,BGV,BoossWojciechowski1993,BaerBallmann2012}. 

This paper aims to present an explicit analysis of the recently introduced $G$-equivariant spectral flow, which takes a value in the representation ring \cite{Fang,IzydorekJanczewskaWaterstraat2021,IzydorekJanczewskaStarostkaWaterstraat2026}, based on a concrete curved Dirac system.
We work on the finite warped cylinder $M=[0,T]\times S^1$ with the twisted Dirac operator on \(M\), building on the warped-cylinder model from \cite{KS}, and focus on \emph{reflection symmetry}. In this model, the circle action is automatic, but the holonomy of the line twist obstructs reflection symmetry. Thus, the first question is when the reflection map
\begin{equation}
r \colon M \to M, \qquad (t,\theta)\mapsto (t,-\theta),
\end{equation}
lifts to a symmetry of the twisted Dirac setting.

For a line twist with holonomy parameter \(A\), equivalently twist class \([A]\in \R/\Z\), the answer is arithmetic:
\begin{equation}
[A]=[-A]
\qquad\Longleftrightarrow\qquad
2A\in\Z.
\end{equation}
Hence, the reflection-compatible case is exactly the half-integral holonomy case. Once the lift exists, reflection pairs the Fourier modes by
\begin{equation}
k\longmapsto k^\vee=-k-2A,
\qquad\text{equivalently}\qquad
m=k+A\longmapsto -m,
\end{equation}
and this pairing governs the APS boundary structure.

This gives a concrete setting in which one can see simultaneously the reflection lift, the holonomy obstruction, the induced Fourier-block pairing, the symmetry-visible APS harmonic data, the \(RO(O(2))\)-valued refinement of spectral flow, and the residual mod-two invariant that remains once pointwise equivariance is lost. Here, the relevant compact group is \(O(2)\), generated by circle rotations and reflection, and the warped-cylinder model provides an explicit bridge between abstract APS/equivariant-spectral-flow theory and computable Dirac boundary operators with symmetry.

We note several related recent directions that are not used directly in the proofs below, but place the current paper in a broader modern context. For equivariant spectral flow and boundary $\eta$-invariants on manifolds with boundary, see Lim and Wang \cite{LimWang2021EquivariantSpectralFlowEtaBoundary}; for broader family-level equivariant spectral-flow frameworks, see Hochs and Yanes \cite{HochsYanes2024EquivariantSpectralFlowDiracType}. A wider equivariant APS background is provided by Hochs, Wang, and Wang \cite{HochsWangWang2019EquivariantAPSI,HochsWangWang2020EquivariantAPSII}, while related boundary-equivariant Toeplitz index questions are studied in \cite{LimWang2021EquivariantToeplitzBoundary}. We also mention recent perspectives on spectral flow and APS index theory in \cite{vanDenDungen2024DiracSchrodingerIndexSpectralFlow,BaerZiemke2025SpectralFlowAPSIndex}, as well as a further $\mathbb Z_2$-valued spectral-flow viewpoint in \cite{BravermanHajSaeediSadegh2024Z2ToeplitzSpectralFlow}. From the viewpoint of condensed-matter physics, reflection and more general crystalline symmetries play a central role in topological crystalline phases; see, for example, a review by Ando and Fu \cite{Ando_Fu_2015}.

The analysis is split into three cases. In the reflection-compatible fixed-holonomy line-twist case, the APS operator decomposes into reflection-paired Fourier blocks. As a consequence, in the invertible-boundary fixed-holonomy case, the ordinary APS index of the flat line-twist model vanishes. When the self-paired sector is present, the finite-dimensional kernel correction has to be kept separately. In either case, a symmetry-visible quantity remains, namely the trace of the lifted reflection on the APS harmonic space. This trace is supported in the self-paired sector \(m=0\).

For moving families, there are again two distinct possibilities. In the line-twist model, if the holonomy is fixed and gauge-trivial, then after choosing the gauge \(A_0=0\), the resulting APS family is genuinely pointwise \(O(2)\)-equivariant, and the relevant invariant is the \(RO(O(2))\)-valued equivariant spectral flow. This may be viewed as a representation-ring refinement of ordinary scalar spectral flow; related \(K\)-theoretic reformulations of Dirac indices via spectral flow and \(\eta\)-invariants, including recent APS and mod-two extensions, appear in \cite{AokiFukayaFurutaMatsuoOnogiYamaguchi2024IndexKTheory, AokiFukayaFurutaMatsuoOnogiYamaguchi2025EtaWilsonIndex,AokiFujitaFukayaFurutaMatsuoOnogiYamaguchi2026CapturingAPS,AokiFujitaFukayaFurutaMatsuoOnogiYamaguchi2026Generalization}. Away from the self-paired sector, every non-self-paired reflection orbit then contributes an even amount to the ordinary spectral flow. 

By contrast, for non-constant holonomy deformations in the original line-twist model, pointwise \(O(2)\)-equivariance necessarily breaks down along the path. The full \(RO(O(2))\)-valued spectral flow is then unavailable for the family, even when the endpoints admit a reflection lift. What remains is a sign-level \(\mathbb Z_2\)-invariant of the deformation: the parity of the APS crossing events \(k+A(s)=0\). Using the kernel-crossing mechanism established in \cite{KS}, we identify this invariant as a \(\mathbb Z_2\)-valued parity, equivalently, the parity of the events at which one Fourier label changes its APS sign assignment. This is the corresponding sign-level replacement for the full equivariant spectral flow in the non-pointwise-\(O(2)\)-equivariant case. See~\cite{DollSchulzBaldesWaterstraat2019,NicolaescuOT}.

This mod-two invariant is conceptually related to mod-two APS index constructions \cite{FukayaFurutaMatsukiMatsuoOnogiYamaguchiYamashita2022}.%
\footnote{This construction provides the mod-two index based on the domain-wall fermion. There is another formulation of the mod-two index based on the overlap fermion~\cite{Clancy:2023ino,Kimura:2023kqs}.}
In the present paper, however, we do not construct a domain-wall operator or a determinant-line transport for the warped-cylinder holonomy family. Consequently, this comparison is used only as motivation for the sign-level nature of the spectral flow $\operatorname{sf}_{\mathbb Z_2}(A)$ and does not play a role in the proofs below.

\paragraph{Main results.}
We summarize the main results as follows.
\begin{enumerate}[label=(\roman*)]
\item The reflection \(r : (t,\theta) \mapsto (t,-\theta)\) lifts to a unitary symmetry of the twisted warped-cylinder Dirac setting if and only if
\begin{equation}
2A\in\Z.
\end{equation}

\item In the reflection-compatible fixed-holonomy case, reflection pairs the Fourier blocks by
\begin{equation}
k\leftrightarrow k^\vee=-k-2A,
\qquad
m\leftrightarrow -m.
\end{equation}
The paired APS blocks are unitarily equivalent. In the invertible-boundary case, the ordinary APS index of the fixed-holonomy line-twist model with flat twisting connection vanishes; when the self-paired sector is present, the finite-dimensional kernel correction has to be kept separately. The reflection trace in the APS harmonic space is supported in the self-paired sector \(m=0\).

\item For gauge-trivial fixed-holonomy bulk perturbation families in the line-twist model, equivalently for $[A_0]=0$ and after choosing the gauge $A_0=0$, the relevant invariant is the \(RO(O(2))\)-valued equivariant spectral flow. Away from the self-paired sector, every non-self-paired reflection orbit contributes an even amount to the ordinary spectral flow.

\item For non-constant holonomy deformations in the original line-twist model, the remaining invariant is the mod-two crossing parity
\begin{equation}
\operatorname{sf}_{\mathbb Z_2}(A),
\end{equation}
namely, the parity of the crossing events \(k+A(s_*)=0\), equivalently, the parity of the one Fourier label changes its APS sign assignment.
\end{enumerate}

Taken together, these results divide the paper into three cases: the reflection-compatible fixed-holonomy line-twist case, the gauge-trivial fixed-holonomy moving-family case in the line-twist model, and the non-constant holonomy-deformation case.  Sections~\ref{sec:O2lift}-\ref{sec:APSindex} treat the first case, Section~\ref{sec:O2_sf} treats the second, and Section~\ref{sec:Z2_flow_1} treats the third.

\paragraph{Organization of the paper.}
In~\autoref{sec:setup} we fix the warped-cylinder Dirac model, its Fourier decomposition, and the APS boundary operator.  \autoref{sec:O2lift} constructs the lifted reflection operator and proves the holonomy obstruction theorem \(2A\in\mathbb Z\).  \autoref{sec:rigorousSF} studies the resulting static
fixed-holonomy APS structure, and \autoref{sec:APSindex} proves the localization of the reflection trace to the self-paired sector.  \autoref{sec:O2_sf} analyzes the gauge-trivial fixed-holonomy line-twist case, precisely \([A_0]=0\) and the
gauge choice \(A_0=0\), and derives the induced \(RO(O(2))\)-valued spectral-flow decomposition.  Finally, \autoref{sec:Z2_flow_1} returns to the original line-twist model and identifies the \(\mathbb Z_2\) crossing parity that survives
once pointwise \(O(2)\)-equivariance is lost.  \autoref{Appendix:A} contains the numerical single-mode APS illustration.

\paragraph{Acknowledgments.}
This work was supported by EIPHI Graduate School (No.~ANR-17-EURE-0002) and the Bourgogne-Franche-Comté region.


\section{Geometric and analytic setup}\label{sec:setup}

This section fixes the warped-cylinder Dirac model and the notation used throughout. The geometric and analytic setup is adapted from \cite{KS}.

\subsection{Geometry, bundles, and function spaces}
Let \(M\) be the finite warped cylinder
\begin{equation}\label{eq:MT_metric}
M=[0,T]\times S^1,
\qquad
g=dt^2+f(t)^2\,d\theta^2,
\end{equation}
where \(f\in C^\infty([0,T])\) is strictly positive. Thus, \(t\in[0,T]\) is the longitudinal coordinate, \(\theta\) is the angular coordinate on \(S^1\), and the metric \(g\) is the warped product line element \(ds^2=dt^2+f(t)^2\,d\theta^2\), so that the circle fiber at \(t\) has radius \(f(t)\).
Its boundary is
\begin{equation}\label{eq:boundary_decomp}
\partial M=Y_0\sqcup Y_T,
\qquad
Y_0=\{0\}\times S^1,
\qquad
Y_T=\{T\}\times S^1.
\end{equation}

We use the orthonormal coframe
\begin{equation}\label{eq:coframe}
e^1=dt,
\qquad
e^2=f(t)\,d\theta,
\end{equation}
with dual frame
\begin{equation}\label{eq:dual_frame}
e_1=\partial_t,
\qquad
e_2=\frac{1}{f(t)}\partial_\theta.
\end{equation}
We orient \(M\) by \(dt\wedge d\theta\). With this convention, the inward unit normal is
\begin{equation}\label{eq:inward_normal}
N=
\begin{cases}
+\partial_t,& \text{on }Y_0,\\
-\partial_t,& \text{on }Y_T.
\end{cases}
\end{equation}

Fix a spin structure on \(M\), and let $S\to M$ denote the corresponding complex spinor bundle. In dimension two, \(S\) has rank \(2\) and splits chirally as
\begin{equation}\label{eq:chirality_split}
S=S^+\oplus S^-.
\end{equation}
Let \(E\to M\) be a Hermitian complex line bundle with a unitary connection
\begin{equation}\label{eq:connectionE}
\nabla^E=d+iA\,d\theta,
\qquad
A\in\mathbb R.
\end{equation}
Only the class
\begin{equation}\label{eq:A_class}
[A]\in \mathbb R/\mathbb Z
\end{equation}
is gauge invariant: if one changes the local trivialization of the line bundle by a unitary gauge transformation, then the parameter \(A\) may change by an integer. However, its class modulo \(\mathbb Z\) remains unchanged. Equivalently, different real numbers \(A\) and \(A+n\), with \(n\in\mathbb Z\), determine gauge-equivalent connections and hence the same holonomy around the circle. Thus, the twisted Dirac operator acts on sections of the twisted bundle,

\begin{equation}\label{eq:twisted_bundle}
S\otimes E\to M.
\end{equation}
From this point onward, we suppress the twist \(E\) from the notation when no confusion can arise, and write \(S\), \(S^\pm\), and \(S_{Y_{t_0}}\) for the corresponding twisted bundles.

For any vector bundle \(V\to M\), we write \(\Gamma(V)\) for the smooth sections of \(V\). We use the standard Hilbert and Sobolev spaces
\begin{equation}\label{eq:bulk_spaces}
L^2(M;V),
\qquad
H^1(M;V),
\end{equation}
defined using the Riemannian volume form of \(g\) and the Hermitian bundle metric. On each boundary component \(Y_{t_0}\), where \(t_0\in\{0,T\}\), we similarly write
\begin{equation}\label{eq:boundary_spaces}
L^2(Y_{t_0};V|_{Y_{t_0}}),
\qquad
H^1(Y_{t_0};V|_{Y_{t_0}}),
\qquad
H^{1/2}(Y_{t_0};V|_{Y_{t_0}})
\end{equation}
for the corresponding boundary spaces, where \(H^{1/2}\) is the standard trace space. In particular, for the twisted spinor bundle \(S\), the trace map takes the form
\begin{equation}\label{eq:trace_map}
H^1(M;S)\longrightarrow H^{1/2}(Y_{t_0};S|_{Y_{t_0}}).
\end{equation}

\subsection{Clifford conventions and the twisted Dirac operator}

Choose Pauli matrices \(\sigma_1,\sigma_2,\sigma_3\), and set
\begin{equation}\label{eq:gamma_def}
\gamma_1=\sigma_1,
\qquad
\gamma_2=\sigma_2,
\qquad
\gamma_3=i\gamma_1\gamma_2=-\sigma_3.
\end{equation}
Thus the chirality decomposition \eqref{eq:chirality_split} is the \(\pm1\)-eigenspace decomposition of \(\gamma_3\).

Using the tensor product connection $
\nabla^{S\otimes E}=\nabla^S\otimes 1 + 1\otimes \nabla^E$, we define the twisted Dirac operator
\begin{equation}\label{eq:Dirac_def}
D:\Gamma(S)\longrightarrow \Gamma(S).
\end{equation}
In the orthonormal frame \eqref{eq:dual_frame}, one obtains
\begin{equation}\label{eq:Dexplicit}
D
=
i\sigma_1\Bigl(\partial_t+\frac{f'(t)}{2f(t)}\Bigr)
+
i\sigma_2\,\frac{1}{f(t)}\bigl(\partial_\theta+iA\bigr).
\end{equation}
Equivalently,
\begin{equation}\label{eq:Dmatrix}
D
=
i
\begin{pmatrix}
0 &
\partial_t+\dfrac{f'(t)}{2f(t)}-\dfrac{i}{f(t)}(\partial_\theta+iA)
\\
\partial_t+\dfrac{f'(t)}{2f(t)}+\dfrac{i}{f(t)}(\partial_\theta+iA)
&
0
\end{pmatrix}.
\end{equation}

\subsection{Fourier reduction and mode spaces}

Since the metric and connection are \(\theta\)-independent, the operator \(D\) commutes with rotations in the \(\theta\)-variable and therefore preserves the Fourier-mode decomposition in the angular direction. Equivalently, \(D\) acts diagonally with respect to the Fourier basis, so the full operator splits into independent mode operators, one for each allowed angular mode. Depending on the chosen spin structure on \(S^1\), the allowed Fourier modes are
\begin{equation}\label{eq:Kdef}
\mathcal K=\mathbb Z
\quad\text{(periodic case)},
\qquad
\mathcal K=\mathbb Z+\frac12
\quad\text{(anti-periodic case)}.
\end{equation}
For \(k\in\mathcal K\), we write
\begin{equation}\label{eq:fourier_ansatz}
\psi(t,\theta)=e^{ik\theta}\binom{u(t)}{v(t)}.
\end{equation}
Since \((\partial_\theta+iA)e^{ik\theta} = i(k+A)e^{ik\theta}\), it is convenient to introduce the shifted mode parameter
\begin{equation}\label{eq:mdef}
m=k+A.
\end{equation}

For each mode \(k\), the natural radial Hilbert and Sobolev spaces are
\begin{equation}\label{eq:mode_spaces}
H_k=L^2([0,T],f(t)\,dt;\mathbb C^2),
\qquad
W_k=H^1([0,T],f(t)\,dt;\mathbb C^2).
\end{equation}
In the \(k\)-th mode, \(D\) reduces to 
\begin{equation}\label{eq:Dk}
D^{(k)}
=
i
\begin{pmatrix}
0 & A^+\\
A^- & 0
\end{pmatrix},\qquad
A^\pm=\partial_t+\frac{f'(t)}{2f(t)}\pm \frac{m}{f(t)}.
\end{equation}
Hence, the eigenvalue equation
\begin{equation}\label{eq:eigen_eq}
D\psi=\lambda\psi,
\qquad
\lambda\in\mathbb R,
\end{equation}
becomes the modewise first-order system, \(A^+v=-i\lambda u\), \(A^-u=-i\lambda v\). 
Equivalently,
\begin{equation}\label{eq:mode_system_components}
v'(t)+\frac{f'(t)}{2f(t)}v(t)+\frac{m}{f(t)}v(t)=-i\lambda u(t),
\qquad
u'(t)+\frac{f'(t)}{2f(t)}u(t)-\frac{m}{f(t)}u(t)=-i\lambda v(t).
\end{equation}
Eliminating one component yields a second-order scalar equation; in the present paper, the key point is that the boundary conditions and spectral data are organized entirely in terms of \(m\).

\subsection{Boundary operators and APS boundary conditions}

Let \(Y_{t_0}\) be one of the boundary circles. The induced boundary spinor bundle is \(S_{Y_{t_0}}=S|_{Y_{t_0}}\).
For the unit tangent vector \(U=e_2|_{t=t_0}\), define the induced boundary Clifford action by
\begin{equation}\label{eq:boundary_clifford}
c(U)=-\,i\,\gamma(N)\gamma(U).
\end{equation}
With the inward normal convention \eqref{eq:inward_normal}, this gives
\begin{equation}\label{eq:cU_formula}
c(U)=
\begin{cases}
\ \sigma_3,& \text{on }Y_0,\\
-\sigma_3,& \text{on }Y_T.
\end{cases}
\end{equation}
The intrinsic tangential operator is
\begin{equation}\label{eq:Dt0}
D_{t_0}
=
\frac{1}{f(t_0)}(-i\partial_\theta+A),
\end{equation}
and the self-adjoint boundary Dirac operators entering the APS projector are
\begin{equation}\label{eq:B0BT}
B_0=\frac{1}{f(0)}\,\sigma_3\,(-i\partial_\theta+A),
\qquad
B_T=-\frac{1}{f(T)}\,\sigma_3\,(-i\partial_\theta+A).
\end{equation}
On the \(k\)-th Fourier mode, these reduce to
\begin{equation}\label{eq:Bk_mode}
B_0^{(k)}=\frac{m}{f(0)}\sigma_3,
\qquad
B_T^{(k)}=-\frac{m}{f(T)}\sigma_3.
\end{equation}

Let \(P_{>0}(B_{t_0})\) denote the orthogonal projection onto the strictly positive spectral subspace of \(B_{t_0}\). For \(m\neq 0\), the APS boundary condition is
\begin{equation}\label{eq:APSabstract}
P_{>0}(B_0)\bigl(\psi|_{Y_0}\bigr)=0,
\qquad
P_{>0}(B_T)\bigl(\psi|_{Y_T}\bigr)=0.
\end{equation}
Since \eqref{eq:Bk_mode} is diagonal, this becomes, modewise,
\begin{subequations}
\begin{equation}\label{eq:APS_mpos}
m>0
\quad\Longrightarrow\quad
u(0)=0,
\qquad
v(T)=0,
\end{equation}
\begin{equation}\label{eq:APS_mneg}
m<0
\quad\Longrightarrow\quad
v(0)=0,
\qquad
u(T)=0.
\end{equation}
\end{subequations}
The global APS convention handles the kernel \(m=0\) stated in Remark~\ref{rem:global_APS_convention}.

Accordingly, the APS domain of the mode operator \(D^{k}\) is
\begin{equation}\label{eq:DomDkAPS}
\Dom(D_{k,\text{APS}})
=
\begin{cases}
\displaystyle\left\{\binom{u}{v}\in W_k:\ u(0)=0,\ v(T)=0\right\},& m>0,\\[1em]
\displaystyle\left\{\binom{u}{v}\in W_k:\ v(0)=0,\ u(T)=0\right\},& m<0.
\end{cases}
\end{equation}

\begin{remark}[APS convention and the self-paired sector]
\label{rem:global_APS_convention}
Throughout the paper, the APS condition is understood in the following sense. On every
nonzero boundary mode, we impose the standard APS projection
\(P_{>0}(B_t)=1_{(0,\infty)}(B_t)\).

If \(B_t\) has a nontrivial kernel, then the APS boundary condition is understood to be completed
by a fixed self-adjoint kernel condition on \(\ker B_t\). In the reflection-symmetric parts of the
paper, this kernel completion is assumed to be invariant under the boundary reflection, denoted later by \(\mathcal R_t\).
Equivalently, the completed APS projector will still be denoted by \(P_t^{\APS}\), and it is assumed
to define a self-adjoint Fredholm APS realization of the Dirac operator.

For \(m\neq 0\), the boundary operators are invertible, so no kernel completion is needed. In
those sectors, the APS condition reduces exactly to
\begin{subequations}
\begin{equation}
m>0
\quad\Longrightarrow\quad
u(0)=0,\qquad v(T)=0,
\end{equation}
and
\begin{equation}
m<0
\quad\Longrightarrow\quad
v(0)=0,\qquad u(T)=0.
\end{equation}
\end{subequations}
Only the self-paired sector \(m=0\) depends on the auxiliary kernel completion. With this convention fixed, we continue to write \(P_t^{APS}\) and \(D^{APS}\) without explicitly
recording the auxiliary kernel completion in the notation.
\end{remark}

\section{$O(2)$-compatible reflection lift: reflection symmetry and the holonomy obstruction}
\label{sec:O2lift}

In this section, we analyze the symmetry of the twisted Dirac problem beyond the automatic
\(S^1\)-rotation symmetry of the warped cylinder. Throughout, \(D\) denotes the twisted Dirac
operator from \autoref{sec:setup}, acting on sections of the twisted spinor bundle \(S\),
with twist class
\begin{equation}\label{eq:twist_class_repeat}
[A]\in \mathbb R/\mathbb Z.
\end{equation}
The rotation symmetry is present for every holonomy class \eqref{eq:twist_class_repeat},
whereas the reflection symmetry is compatible with the twisting only when
\begin{equation}\label{eq:2A_intro}
2A\in\mathbb Z.
\end{equation}

\subsection{The rotation symmetry}

Let
\begin{equation}\label{eq:rotation_map}
R_\varphi:M\longrightarrow M,
\qquad
(t,\theta) \mapsto(t,\theta+\varphi),
\qquad
\varphi\in\mathbb R/2\pi\mathbb Z,
\end{equation}
be the standard rotation of the angular variable. Since both the warped metric $g=dt^2+f(t)^2\,d\theta^2$ and the connection form $ iA\,d\theta $
are independent of \(\theta\), the operator \(D\) is invariant under this action.

We denote by
\begin{equation}\label{eq:rotation_rep}
\rho(\varphi):\Gamma(S)\longrightarrow \Gamma(S)
\end{equation}
the induced unitary action on sections, given in the chosen trivialization by pullback:
\begin{equation}\label{eq:rotation_pullback}
(\rho(\varphi)\psi)(t,\theta)=\psi(t,\theta+\varphi).
\end{equation}
Then
\begin{equation}\label{eq:D_rotation_commute}
\rho(\varphi)\,D\,\rho(\varphi)^{-1}=D.
\end{equation}
Equivalently, \(D\) commutes with the circle action and therefore decomposes into Fourier modes, as in~\autoref{sec:setup}.

\subsection{The reflection map and its effect on the angular operator}

We now introduce the reflection
\begin{equation}\label{eq:reflection_map}
r:M\longrightarrow M,
\qquad
(t,\theta) \longmapsto (t,-\theta).
\end{equation}
This map is an isometry of the warped metric, but it reverses the angular coordinate.

In the chosen trivialization, we write $r^* $
for pullback on coefficient functions induced by \(r\). Since \(r\) reverses \(\theta\), one has
\begin{equation}\label{eq:reflection_on_angular_momentum}
r^*\,(-i\partial_\theta)\,(r^*)^{-1}=+\,i\partial_\theta.
\end{equation}
Thus, reflection reverses the angular momentum operator.

It is convenient to write the angular part of \(D\) using
\begin{equation}\label{eq:Ta_def}
T_A=-i\partial_\theta+A.
\end{equation}
Then the explicit formula for the twisted Dirac operator from \autoref{sec:setup} becomes
\begin{equation}\label{eq:D_as_TA}
D
=
i\sigma_1\Bigl(\partial_t+\frac{f'(t)}{2f(t)}\Bigr)
-
\frac{1}{f(t)}\,\sigma_2\,T_A.
\end{equation}
Under reflection, one has
\begin{equation}\label{eq:rstar_TA}
r^*\,T_A\,(r^*)^{-1}
=
-\,T_{-A}.
\end{equation}
Indeed, since \(r(t,\theta)=(t,-\theta)\) acts only on the angular variable, it reverses \(\partial_\theta\) but acts trivially on the constant scalar parameter \(A\). Thus,
\begin{equation}\label{eq:rstar_TA_calc}
r^*\,T_A\,(r^*)^{-1}
=
r^*(-i\partial_\theta+A)(r^*)^{-1}
=
+i\partial_\theta+A
=
-(-i\partial_\theta-A)
=
-\,T_{-A}.
\end{equation}

Under the reflection \eqref{eq:reflection_map}, the connection term \(A\,d\theta\) is sent to \(-A\,d\theta\), so the holonomy class transforms as
\begin{equation}\label{eq:holonomy_reflection}
[A]\longmapsto [-A]
\qquad
\text{in }\mathbb R/\mathbb Z.
\end{equation}
Hence, reflection preserves the twisted conditions only when the reflected twist class is gauge-equivalent to the original one.

\subsection{Gauge transformations and the arithmetic condition}

For each integer \(n\in\mathbb Z\), we define
\begin{equation}\label{eq:gn_def}
g_n:S^1\longrightarrow U(1),
\qquad
g_n(\theta)=e^{in\theta}.
\end{equation}
In the chosen trivialization of the line bundle, this acts on sections by multiplication by \(e^{in\theta}\). Conjugating \(T_A\) by \(g_n\) gives
\begin{equation}\label{eq:gauge_shift}
g_n\,T_A\,g_n^{-1}
=
T_{A-n}.
\end{equation}
Indeed,
\begin{equation}\label{eq:gauge_shift_calc}
g_n(-i\partial_\theta+A)g_n^{-1}
=
-i\partial_\theta+(A-n).
\end{equation}
Thus, the parameter \(A\) is defined only modulo integers, and the invariant datum is the class \([A]\in\mathbb R/\mathbb Z\).

Reflection is compatible with the twisting exactly when the original and reflected classes agree:
\begin{equation}\label{eq:AeqminusA}
[A]=[-A]
\qquad
\text{in }\mathbb R/\mathbb Z.
\end{equation}
This is equivalent to the condition
\begin{equation}\label{eq:2Acond}
2A\in\mathbb Z.
\end{equation}
Indeed, \eqref{eq:AeqminusA} means that \(A-(-A)=2A\) is an integer.

\subsection{The spinor reflection and the lifted reflection operator}\label{sec:spinor_reflection_line_reflection}

To lift the base reflection to the twisted spinor bundle, we first choose a matrix acting on the spinor factor. Let
\begin{equation}\label{eq:Ur_def}
U_r=\sigma_1.
\end{equation}
Since the Clifford matrices were fixed in~\autoref{sec:setup} as \eqref{eq:gamma_def}, one computes 
\begin{equation}\label{eq:Ur_gamma}
U_r\gamma_1U_r^{-1}=\gamma_1, \qquad
U_r\gamma_2U_r^{-1}=-\gamma_2.
\end{equation}
Thus, \(U_r\) fixes the radial Clifford generator and reverses the angular one, exactly as required by the reflection \(\theta\mapsto -\theta\).

Assume now that \eqref{eq:2Acond} holds. Then \(-2A\in\mathbb Z\), so the gauge factor
\begin{equation}\label{eq:gminus2A_def}
g_{-2A}(\theta)=e^{-2iA\theta}
\end{equation}
is single-valued on \(S^1\). We define the lifted reflection operator on sections by
\begin{equation}\label{eq:Ucalr_def}
\mathcal U_r=U_r\,g_{-2A}\,r^*.
\end{equation}
Here, \(r^*\) implements the base reflection, \(U_r\) supplies the spinor reflection, and \(g_{-2A}\) corrects the reflected twist \(-A\) back to \(A\).

\begin{remark}[Bundle-theoretic form of the reflection lift]
It is conceptually useful to separate the bundle-level lift from the induced operator on sections. Let \(\pi:S \to M\) denote the twisted spinor bundle projection, and let
\begin{equation}\label{eq:bundle_projection_pir}
\pi_r:r^*(S)\longrightarrow M
\end{equation}
denote the projection of the pullback bundle under the reflection \(r:M\to M\).

Assuming \eqref{eq:2Acond}, define the fiberwise bundle isomorphism
\begin{equation}\label{eq:Phi_r_def}
\Phi_r=U_r\,g_{-2A}:r^*(S)\longrightarrow S.
\end{equation}
Then \(\Phi_r\) is the bundle-level reflection lift associated with \(r\), covering the identity on \(M\), in the sense that
\begin{equation}\label{eq:Phi_r_covers_id}
\pi\circ \Phi_r=\pi_r.
\end{equation}
Equivalently, \(\Phi_r\) fits into the commutative diagram
\begin{equation}
\begin{tikzcd}[column sep=large,row sep=large]
r^*(S) \arrow[r,"\Phi_r"] \arrow[d,"\pi_r"'] &
S \arrow[d,"\pi"] \\
M \arrow[r,"\text{id}_M"'] & M
\end{tikzcd}
\end{equation}
The corresponding operator in sections is
\begin{equation}\label{eq:Ucalr_from_Phi}
\mathcal U_r=\Gamma(\Phi_r)\circ r^*,
\end{equation}
where \(\Gamma(\Phi_r)\) denotes the map induced by \(\Phi_r\) in sections.
\end{remark}

\begin{lemma}[Intertwining of the reflected twisted operator]
\label{lem:reflection_intertwining}
Assume \eqref{eq:2Acond}. Then the operator \(\mathcal U_r\) defined in \eqref{eq:Ucalr_def}
is unitary in \(L^2(M;S)\) and satisfies
\begin{equation}\label{eq:Ucalr_intertwine}
\mathcal U_r\,D\,\mathcal U_r^{-1}=D
\end{equation}
on \(\Gamma(S)\). Equivalently, the following diagram commutes:
\begin{equation}
\begin{tikzcd}[column sep=large,row sep=large]
\Gamma(S) \arrow[r,"\mathcal U_r"] \arrow[d,"D"'] &
\Gamma(S) \arrow[d,"D"] \\
\Gamma(S) \arrow[r,"\mathcal U_r"'] &
\Gamma(S)
\end{tikzcd}
\end{equation}
\end{lemma}

\begin{proof}
The operator \(r^*\) is unitary because \(r\) is an isometry of the warped metric, \(U_r=\sigma_1\) is unitary on the spinor factor, and \(g_{-2A}\) is unitary as a \(U(1)\)-valued multiplication operator. Hence \(\mathcal U_r\) is unitary.

Using \eqref{eq:D_as_TA}, \eqref{eq:rstar_TA}, \eqref{eq:Ur_gamma}, we obtain
\begin{equation}\label{eq:Ur_rstar_conj_D}
U_r\,r^*\,D\,(r^*)^{-1}U_r^{-1}
=
i\sigma_1\Bigl(\partial_t+\frac{f'(t)}{2f(t)}\Bigr)
-
\frac{1}{f(t)}\,\sigma_2\,T_{-A}.
\end{equation}
Next, by \eqref{eq:gauge_shift} with \(n=-2A\),
\begin{equation}\label{eq:gauge_compensation}
g_{-2A}\,T_{-A}\,g_{-2A}^{-1}
=
T_{-A-(-2A)}
=
T_A.
\end{equation}
Therefore
\begin{equation}\label{eq:full_conj_D}
\mathcal U_r\,D\,\mathcal U_r^{-1}
=
i\sigma_1\Bigl(\partial_t+\frac{f'(t)}{2f(t)}\Bigr)
-
\frac{1}{f(t)}\,\sigma_2\,T_A
=
D.
\end{equation}
This proves \eqref{eq:Ucalr_intertwine}.
\end{proof}

\begin{remark}
\label{rem:lifted_reflection_notation}
Whenever \eqref{eq:2Acond} holds, we keep the notation
\begin{equation}\label{eq:Ucalr_repeat}
\mathcal U_r=U_r\,g_{-2A}\,r^*
\end{equation}
for the lifted reflection operator on sections.
\end{remark}

\subsection{The holonomy obstruction to a reflection lift}

We now state the precise obstruction theorem.

\begin{theorem}[Holonomy obstruction to a reflection lift]
\label{thm:O2lift}
Let \(D\) be the twisted warped-cylinder Dirac operator with constant holonomy parameter \(A\).
Then the reflection map
\begin{equation}\label{eq:reflection_map_repeat}
r:(t,\theta)\longmapsto(t,-\theta)
\end{equation}
lifts to a unitary symmetry of the twisted Dirac problem if and only if
\begin{equation}\label{eq:O2_if_and_only_if}
2A\in\mathbb Z.
\end{equation}
\end{theorem}

\begin{proof}
Assume first that \eqref{eq:O2_if_and_only_if} holds. Then the operator \(\mathcal U_r\) from
\eqref{eq:Ucalr_def} is globally defined and unitary, and by Lemma~\ref{lem:reflection_intertwining}
it satisfies
\begin{equation}\label{eq:reflection_symmetry_proof}
\mathcal U_r\,D\,\mathcal U_r^{-1}=D.
\end{equation}
Thus, reflection lifts to a unitary symmetry of the twisted problem.

Conversely, suppose a unitary reflection lift exists. Then conjugation by that lift identifies the reflected twisted operator with the original one. In particular, the reflected connection class \([-A]\in \mathbb R/\mathbb Z\) must define the same twisted problem as the original class \([A]\in \mathbb R/\mathbb Z\). Hence
\begin{equation}\label{eq:proof_holonomy_classes}
[A]=[-A]
\qquad
\text{in }\mathbb R/\mathbb Z.
\end{equation}
By \eqref{eq:AeqminusA}, this is equivalent to
\begin{equation}\label{eq:proof_2A_integer}
2A\in\mathbb Z.
\end{equation}
This proves the theorem.
\end{proof}

\subsection{Compatibility with rotations}

We now compute the relation between the lifted reflection and the rotation action. Recall that \(\rho(\varphi)\) was defined in \eqref{eq:rotation_rep} by pullback under \(R_\varphi\).

\begin{proposition}[Twisted semidirect relation with rotations]
\label{prop:semidirect_relation}
Assume \eqref{eq:2Acond}. Then
\begin{equation}\label{eq:semidirect_relation}
\mathcal U_r\,\rho(\varphi)\,\mathcal U_r^{-1}
=
e^{-2iA\varphi}\,\rho(-\varphi).
\end{equation}
\end{proposition}

\begin{proof}
Since \(U_r\) is constant, it commutes with \(\rho(\varphi)\). Also, on the base manifold, we have \(r\circ R_\varphi\circ r^{-1}=R_{-\varphi}\), hence \(r^*\,\rho(\varphi)\,(r^*)^{-1}
=
\rho(-\varphi)\).
Therefore
\begin{equation}\label{eq:semidirect_calc_start}
\mathcal U_r\,\rho(\varphi)\,\mathcal U_r^{-1}
=
g_{-2A}\,\rho(-\varphi)\,g_{-2A}^{-1}.
\end{equation}
Now \(g_{-2A}(\theta)=e^{-2iA\theta}\), so
\begin{equation}\label{eq:gauge_rotation_commute}
g_{-2A}(\theta)\,g_{-2A}(\theta-\varphi)^{-1}
=
e^{-2iA\varphi}.
\end{equation}
Hence conjugating \(\rho(-\varphi)\) by \(g_{-2A}\) produces the scalar factor \(e^{-2iA\varphi}\), giving \eqref{eq:semidirect_relation}.
\end{proof}

\begin{lemma}[Involution relation]
\label{lem:Ur_involution}
Assume \eqref{eq:2Acond}. Then
\begin{equation}\label{eq:Ur_squared}
\mathcal U_r^2=\Id.
\end{equation}
\end{lemma}

\begin{proof}
Since \(U_r=\sigma_1\), one has \(U_r^2=\Id\).
Also, \(r^2=\Id\) and \(r^*\,g_{-2A}\,(r^*)^{-1}=g_{2A}\).
Therefore,
\begin{equation}\label{eq:Ur_square_calc}
\mathcal U_r^2
=
U_r g_{-2A} r^* U_r g_{-2A} r^*
=
U_r^2\, g_{-2A}\, g_{2A}\, (r^*)^2
=
\Id.
\end{equation}
\end{proof}

\begin{remark}[Reflection lift versus an honest \texorpdfstring{$O(2)$}{O(2)}-action]
\label{rem:central_character}
The relation in Proposition~\ref{prop:semidirect_relation} shows that, in a fixed
representative \(A\), the lifted reflection and the rotation action satisfy
\begin{equation}
\mathcal U_r\,\rho(\varphi)\,\mathcal U_r^{-1}
=
e^{-2iA\varphi}\rho(-\varphi).
\end{equation}
Thus they satisfy the usual semidirect-product relation for an honest pointwise
\(O(2)\)-action precisely when the scalar factor is trivial for all \(\varphi\). In the
chosen representative this means
\begin{equation}
A=0.
\end{equation}
Equivalently, in gauge-invariant language, the obstruction vanishes precisely when the
holonomy class is gauge-trivial,
\begin{equation}
[A]=0\in \mathbb R/\mathbb Z.
\end{equation}
In that case one may make an integer gauge transformation and work in the representative
\(A=0\).

Thus, Section~\ref{sec:O2lift} gives a reflection lift whenever \(2A\in\mathbb Z\), but the
standard pointwise \(O(2)\)-equivariant spectral-flow framework used later in
Section~\ref{sec:O2_sf} requires the stronger gauge-trivial condition \([A]=0\), together
with the gauge choice \(A=0\). For the sectorwise arguments below, the important point is
that a reflection lift pairs the shifted angular parameter \(m\) with \(-m\).
\end{remark}

\subsection{Mode pairing induced by reflection}\label{sec:Mode_paring_induced_reflection}

Assume \eqref{eq:2Acond}. For a Fourier mode
\begin{equation}
\psi_k(t,\theta)=e^{ik\theta}\binom{u(t)}{v(t)},
\qquad k\in\mathcal K,
\end{equation}
equation \eqref{eq:Ucalr_def} gives
\begin{equation}\label{eq:reflection_on_mode}
\mathcal U_r\psi_k(t,\theta)
=
U_r\,e^{-2iA\theta}\,e^{-ik\theta}\binom{u(t)}{v(t)}
=
e^{-i(k+2A)\theta}\,U_r\binom{u(t)}{v(t)}.
\end{equation}
Hence reflection sends
\begin{equation}\label{eq:kpair}
k\longmapsto k^\vee= -k-2A.
\end{equation}
Because \(2A\in\mathbb Z\), this preserves the Fourier lattice \(\mathcal K\). Writing \(m=k+A\), one obtains
\begin{equation}\label{eq:m_pairing}
m^\vee=k^\vee+A=-m.
\end{equation}
Thus, reflection exchanges the modes labeled by \(m\) and \(-m\).

\begin{proposition}[Reflection exchanges boundary blocks]
\label{prop:blockpair}
Assume \eqref{eq:2Acond}. For each \(k\in\mathcal K\),
\begin{equation}\label{eq:Bk_again}
B_0^{(k)}=\frac{m}{f(0)}\sigma_3,
\qquad
B_T^{(k)}=-\frac{m}{f(T)}\sigma_3
\end{equation}
is paired by reflection with the \(k^\vee\)-block
\begin{equation}\label{eq:Bkvee}
B_0^{(k^\vee)}=-\frac{m}{f(0)}\sigma_3,
\qquad
B_T^{(k^\vee)}=\frac{m}{f(T)}\sigma_3.
\end{equation}
Equivalently, reflection sends \(m\mapsto -m\), so the APS sign choice is reversed.
\end{proposition}

\begin{proof}
Substitute \(m^\vee=-m\) from \eqref{eq:m_pairing} into the boundary formulas
\eqref{eq:Bk_again}. This gives \eqref{eq:Bkvee}. Since the APS boundary condition depends
only on the sign of \(m\), the sign choice is reversed accordingly.
\end{proof}

\section{Static fixed-holonomy APS structure}
\label{sec:rigorousSF}

If reflection is required at every parameter value in the line-twist model, then the
holonomy must be constant. We therefore fix
\begin{equation}\label{eq:sec4_fixed_holonomy_value}
A=A_0,
\qquad
2A_0\in\mathbb Z,
\end{equation}
and record the resulting static consequences: the pairing \(m\leftrightarrow -m\), the
unitary equivalence of the paired APS blocks, and the vanishing of the boundary
\(\eta\)-terms away from \(m=0\). The genuinely moving \(O(2)\)-equivariant family is
treated later in \autoref{sec:O2_sf}, and non-equivariant holonomy deformations in
\autoref{sec:Z2_flow_1}.

\subsection{Pointwise reflection compatibility forces constant holonomy}

\begin{proposition}[Pointwise reflection compatibility forces constant holonomy]
\label{prop:sec4_constant_holonomy}
Let \(A:[0,1]\to\R\) be continuous. If \(2A(s)\in\Z\) for all \(s\in[0,1]\), then \(A\) is constant.
\end{proposition}

\begin{proof}
The image of \(A\) lies in the discrete set \(\frac12\Z\). Since \([0,1]\) is connected and
\(A\) is continuous, \(A([0,1])\) must be connected, hence a singleton.
\end{proof}

\subsection{The fixed-holonomy APS operator}

Fix \(A_0\) as in \eqref{eq:sec4_fixed_holonomy_value}. The twisted Dirac operator is
\begin{equation}\label{eq:sec4_fixed_Dirac}
D_{A_0}
=
i\sigma_1\Bigl(\partial_t+\frac{f'(t)}{2f(t)}\Bigr)
+
i\sigma_2\,\frac{1}{f(t)}(\partial_\theta+iA_0),
\end{equation}
with boundary operators
\begin{equation}\label{eq:sec4_fixed_boundary_ops}
B_{0,A_0}
=
\frac{1}{f(0)}\,\sigma_3(-i\partial_\theta+A_0),
\qquad
B_{T,A_0}
=
-\frac{1}{f(T)}\,\sigma_3(-i\partial_\theta+A_0).
\end{equation}

For \(t\in\{0,T\}\), let \(P_{t,A_0}^{\APS}\) denote the completed APS projector in the sense of
Remark~\ref{rem:global_APS_convention}. Thus, on every nonzero boundary mode it agrees with
\(\mathbf 1_{(0,\infty)}(B_{t,A_0})\), while in the boundary kernel it includes the fixed auxiliary self-adjoint kernel completion. With this convention, the APS domain is
\begin{equation}\label{eq:sec4_fixed_APS_domain}
\Dom(D_{A_0}^{\APS})
=
\Bigl\{
\psi\in H^1(M;S):
P_{0,A_0}^{\APS}(\psi|_{Y_0})=0,\ 
P_{T,A_0}^{\APS}(\psi|_{Y_T})=0
\Bigr\}.
\end{equation}

Setting \(m=k+A_0\), the modewise boundary operators are
\begin{equation}\label{eq:sec4_fixed_boundary_mode_ops}
B_{0,A_0}^{(k)}=\frac{m}{f(0)}\sigma_3,
\qquad
B_{T,A_0}^{(k)}=-\frac{m}{f(T)}\sigma_3.
\end{equation}
For \(m\neq 0\), the APS condition becomes
\begin{subequations}\label{eq:sec4_fixed_APS_bc}
\begin{equation}\label{eq:sec4_fixed_APS_bc_positive}
m>0
\quad\Longrightarrow\quad
u(0)=0,\qquad v(T)=0,
\end{equation}
\begin{equation}\label{eq:sec4_fixed_APS_bc_negative}
m<0
\quad\Longrightarrow\quad
v(0)=0,\qquad u(T)=0.
\end{equation}
\end{subequations}
The self-paired sector is \(m=0\).

\subsection{Reflection on boundary blocks}

Because \(2A_0\in\Z\), the lifted reflection is
\(\mathcal U_r=U_r\,g_{-2A_0}\,r^*\), with boundary restriction
\begin{equation}\label{eq:sec4_boundary_reflection_operator_fixed}
\mathcal R_t:L^2(Y_t;S)\to L^2(Y_t;S).
\end{equation}
For each \(k\in\mathcal K\), define
\begin{equation}
k^\vee=-k-2A_0.
\end{equation}
Then
\begin{equation}
m^\vee=k^\vee+A_0=-m,
\end{equation}
as discussed in \autoref{sec:Mode_paring_induced_reflection}.

\begin{lemma}[Reflection exchanges the boundary labels]
\label{lem:sec4_fixed_boundary_pairing}
For each \(k\in\mathcal K\),
\begin{equation}
B_{0,A_0}^{(k^\vee)}=-B_{0,A_0}^{(k)},
\qquad
B_{T,A_0}^{(k^\vee)}=-B_{T,A_0}^{(k)}.
\end{equation}
\end{lemma}

\begin{proof}
Insert \(m^\vee=-m\) into
\begin{equation}
B_{0,A_0}^{(k)}=\frac{m}{f(0)}\sigma_3,
\qquad
B_{T,A_0}^{(k)}=-\frac{m}{f(T)}\sigma_3.
\end{equation}
\end{proof}

\begin{lemma}[Reflection conjugates boundary blocks]
\label{lem:sec4_reflection_conjugates_boundary_blocks}
For each \(k\in\mathcal K\) and \(t\in\{0,T\}\), one has
\begin{equation}
\mathcal R_t B^{(k)}_{t,A_0}\mathcal R_t^{-1}
=
B^{(k^\vee)}_{t,A_0}.
\end{equation}
\end{lemma}

\begin{proof}
The boundary operators are
\begin{equation}
B^{(k)}_{0,A_0}
=
\frac{m}{f(0)}\sigma_3,
\qquad
B^{(k)}_{T,A_0}
=
-\frac{m}{f(T)}\sigma_3.
\end{equation}
The lifted reflection changes the Fourier label from \(k\) to \(k^\vee=-k-2A_0\), so
\begin{equation}
m^\vee=-m.
\end{equation}
On the spinor factor, the reflection matrix is \(U_r=\sigma_1\), and
\begin{equation}
\sigma_1\sigma_3\sigma_1=-\sigma_3.
\end{equation}
Therefore, at \(Y_0\),
\begin{equation}
\mathcal R_0 B^{(k)}_{0,A_0}\mathcal R_0^{-1}
=
\frac{m}{f(0)}\sigma_1\sigma_3\sigma_1
=
-\frac{m}{f(0)}\sigma_3
=
\frac{m^\vee}{f(0)}\sigma_3
=
B^{(k^\vee)}_{0,A_0}.
\end{equation}
Similarly, at \(Y_T\),
\begin{equation}
\mathcal R_T B^{(k)}_{T,A_0}\mathcal R_T^{-1}
=
-\frac{m}{f(T)}\sigma_1\sigma_3\sigma_1
=
\frac{m}{f(T)}\sigma_3
=
-\frac{m^\vee}{f(T)}\sigma_3
=
B^{(k^\vee)}_{T,A_0}.
\end{equation}
\end{proof}

\begin{corollary}[Action on APS projections]
\label{cor:sec4_fixed_projector_pairing}
Assume \(m\neq 0\). Then, for \(t\in\{0,T\}\),
\begin{equation}
\mathcal R_t\,P_{>0}\!\bigl(B_{t,A_0}^{(k)}\bigr)\,\mathcal R_t^{-1}
=
P_{>0}\!\bigl(B_{t,A_0}^{(k^\vee)}\bigr).
\end{equation}
\end{corollary}

\begin{proof}
This follows from Lemma~\ref{lem:sec4_reflection_conjugates_boundary_blocks} by functional calculus.
\end{proof}

\subsection{Paired APS blocks}

Let
\begin{equation}\label{eq:sec4_fixed_L2_decomp}
L^2(M;S)
=
\bigoplus_{k\in\mathcal K}\mathcal H_k,
\qquad
\mathcal H_k\cong L^2([0,T],f(t)\,dt;\C^2),
\end{equation}
be the Fourier decomposition. For \(m\neq 0\), set \(\mathcal H_{\{m,-m\}} = \mathcal H_k\oplus\mathcal H_{k^\vee}\). On this block, we have
\begin{equation}\label{eq:sec4_fixed_block_form}
D_{A_0}^{\APS}\big|_{\mathcal H_{\{m,-m\}}}
=
\begin{pmatrix}
D_{A_0,m}^{\APS} & 0\\
0 & D_{A_0,-m}^{\APS}
\end{pmatrix}.
\end{equation}

\begin{proposition}[Unitary equivalence of paired APS blocks]
\label{lem:pairunitary}
For each \(m\neq 0\), the operators \(D_{A_0,m}^{\APS}\) and \(D_{A_0,-m}^{\APS}\) are unitarily equivalent. Consequently,
\begin{equation}\label{eq:sec4_fixed_spectral_symmetry}
\Spec\bigl(D_{A_0,m}^{\APS}\bigr)
=
\Spec\bigl(D_{A_0,-m}^{\APS}\bigr)
\end{equation}
with the same multiplicities.
\end{proposition}

\begin{proof}
The lifted reflection \(\mathcal U_r\) maps the \(k\)-block unitarily to the \(k^\vee\)-block
and transports the APS condition by
Corollary~\ref{cor:sec4_fixed_projector_pairing}. Hence \(\mathcal{U}_r\,D_{A_0,m}^{\APS}\,\mathcal{U}_r^{-1} = D_{A_0,-m}^{\APS}\).
\end{proof}

\begin{remark}\label{rem:sec4_evenness_bridge}
The significance of Proposition~\ref{lem:pairunitary} is structural: any later fixed-holonomy perturbation that preserves the same reflection symmetry inherits identical contributions from the paired sectors \(m\) and \(-m\). Away from the self-paired sector \(m=0\), spectral-flow contributions therefore occur in equal pairs.
\end{remark}

\subsection{Vanishing of boundary \texorpdfstring{$\eta$}{eta}-terms off the self-paired sector}

For any self-adjoint operator \(B\) with discrete spectrum, we write
\begin{equation}\label{eq:eta_definition_local}
\eta(B,s)=\sum_{\lambda\in\Spec(B)\setminus\{0\}}\sign(\lambda)\,|\lambda|^{-s},
\qquad
\Re(s)\gg 0,
\end{equation}
and denote by
\begin{equation}\label{eq:eta_at_zero_local}
\eta(B)=\eta(B,0)
\end{equation}
its APS \(\eta\)-invariant, defined by meromorphic continuation to \(s=0\). We also write
\begin{equation}\label{eq:h_definition_local}
h(B)=\dim\ker(B),
\end{equation}
and
\begin{equation}\label{eq:reduced_eta_definition_local}
\bar\eta(B)=\frac{\eta(B)+h(B)}{2}
\end{equation}
for the reduced \(\eta\)-invariant.

For each Fourier mode \(k\in\mathcal K\), the boundary operators are
\begin{equation}\label{eq:sec4_fixed_eta_boundary_blocks}
B_{0,A_0}^{(k)}=\frac{m}{f(0)}\sigma_3,
\qquad
B_{T,A_0}^{(k)}=-\frac{m}{f(T)}\sigma_3,
\qquad
m=k+A_0.
\end{equation}
If \(m\neq 0\), then each boundary block is invertible and has spectrum symmetric about \(0\):
\begin{equation}\label{eq:sec4_fixed_eta_block_spectrum}
\Spec\bigl(B_{0,A_0}^{(k)}\bigr)
=
\left\{\frac{m}{f(0)},-\frac{m}{f(0)}\right\},
\qquad
\Spec\bigl(B_{T,A_0}^{(k)}\bigr)
=
\left\{-\frac{m}{f(T)},\frac{m}{f(T)}\right\}.
\end{equation}
Therefore,
\begin{equation}\label{eq:sec4_fixed_eta_sectorwise_zero}
\eta\!\bigl(B_{0,A_0}^{(k)}\bigr)=0,
\qquad
\eta\!\bigl(B_{T,A_0}^{(k)}\bigr)=0
\qquad (m\neq 0).
\end{equation}
Since the kernel is trivial for \(m\neq 0\),
\begin{equation}\label{eq:sec4_fixed_kernel_zero_nonwall}
h\!\bigl(B_{0,A_0}^{(k)}\bigr)=0,
\qquad
h\!\bigl(B_{T,A_0}^{(k)}\bigr)=0
\qquad (m\neq 0),
\end{equation}
and hence also
\begin{equation}\label{eq:sec4_fixed_reduced_eta_sectorwise_zero}
\bar\eta\!\bigl(B_{0,A_0}^{(k)}\bigr)=0,
\qquad
\bar\eta\!\bigl(B_{T,A_0}^{(k)}\bigr)=0
\qquad (m\neq 0).
\end{equation}

Thus, every non-self-paired sector contributes trivially to the APS boundary correction on each boundary component separately. However, in the self-paired sector \(m=0\), the boundary operator may have nontrivial kernel, and the reduced invariant
\(\bar\eta=(\eta+h)/2\) need not vanish. Hence \(m=0\) is the unique sector in which a nontrivial boundary contribution can occur.

\begin{remark}\label{rem:sec4_selfpaired_sector}
The self-paired sector \(m=0\) occurs exactly when \(k=-A_0\). It is the unique sector fixed by the reflection pairing \(m\leftrightarrow -m\), but in the present explicit model, the stronger statement holds that every nonzero boundary block already has vanishing \(\eta\)-invariant on its own because its spectrum is symmetric about \(0\). The genuinely moving-family analysis of the self-paired sector is deferred to \autoref{sec:O2_sf}.
\end{remark}

\subsection{Ordinary APS index in the static line-twist case}

We briefly record the ordinary chiral APS index associated with the fixed-holonomy
line-twist model. This is distinct from the reflection trace studied in \autoref{sec:APSindex}: the latter is the natural symmetry-visible quantity for the orientation-reversing reflection,
whereas the ordinary APS index is the Fredholm index of the positive-chirality half of the
Dirac operator.

Let
\begin{equation}\label{eq:sec4_chiral_half_}
D_{A_0}^{+}:\Gamma(S^{+})\longrightarrow \Gamma(S^{-})
\end{equation}
denote the positive-chirality half of the twisted Dirac operator \(D_{A_0}\), equipped with
the APS boundary condition induced by \(\Dom(D_{A_0}^{\APS})\).

We write
\begin{equation}\label{eq:sec4_chiral_boundary_decomposition}
B_{t,A_0}=
\begin{pmatrix}
B^{-}_{t,A_0} & 0\\
0 & B^{+}_{t,A_0}
\end{pmatrix},
\qquad t\in\{0,T\},
\end{equation}
for the chiral decomposition of the full boundary operator on
\(S|_{Y_t}=S^{-}|_{Y_t}\oplus S^{+}|_{Y_t}\). Since \(S^{-}\) is the upper component and
\(S^{+}\) is the lower component in our convention, the positive-chirality boundary operators
are given modewise by
\begin{equation}\label{eq:sec4_positive_chiral_boundary_modes}
B^{+}_{0,A_0}(k)=-\frac{m}{f(0)},
\qquad
B^{+}_{T,A_0}(k)=+\frac{m}{f(T)},
\qquad
m=k+A_0.
\end{equation}

After identifying the two boundary circles by the angular coordinate, the endpoint
positive-chirality boundary operators satisfy the sign/scale relation
\begin{equation}\label{eq:sec4_positive_chiral_endpoint_relation}
B^{+}_{T,A_0}
=
-\,c\,U\,B^{+}_{0,A_0}\,U^{-1},
\qquad
c=\frac{f(0)}{f(T)}>0,
\end{equation}
with \(U\) unitary.

\begin{proposition}[Vanishing of the ordinary APS index in the invertible fixed-holonomy case]
\label{prop:sec4_static_APS_index_zero}
Assume \(A=A_0\), \(2A_0\in\Z\), and
\begin{equation}
-A_0\notin\mathcal K.
\end{equation}
Equivalently, the self-paired mode \(m=0\) is absent. Then
\begin{equation}
\ind\!\bigl(D_{A_0,\APS}^{+}\bigr)=0.
\end{equation}
\end{proposition}

\begin{proof}
By the APS index theorem, specialized to the present flat line-twist model and written for the
positive-chirality boundary operators, the local index density has no interior contribution. Hence
\begin{equation}
\ind\!\bigl(D_{A_0,\APS}^{+}\bigr)
=
-\bar\eta\!\left(B^{+}_{0,A_0}\right)
-\bar\eta\!\left(B^{+}_{T,A_0}\right).
\end{equation}
By \eqref{eq:sec4_positive_chiral_endpoint_relation}, the endpoint positive-chirality boundary
operators satisfy
\begin{equation}
B^{+}_{T,A_0}
=
-\,c\,U\,B^{+}_{0,A_0}\,U^{-1},
\qquad
c=\frac{f(0)}{f(T)}>0.
\end{equation}
Because \(-A_0\notin\mathcal K\), the boundary operators are invertible. Therefore
\begin{equation}
\bar\eta(B)=\frac{\eta(B)}{2}.
\end{equation}
Positive rescaling and unitary conjugation do not change the eta-invariant, while multiplication
by \(-1\) changes \(\eta\) to \(-\eta\). Hence
\begin{equation}
\bar\eta\!\left(B^{+}_{T,A_0}\right)
=
-\bar\eta\!\left(B^{+}_{0,A_0}\right).
\end{equation}
Therefore
\begin{equation}
\bar\eta\!\left(B^{+}_{0,A_0}\right)
+
\bar\eta\!\left(B^{+}_{T,A_0}\right)
=
0,
\end{equation}
and the result follows.
\end{proof}

\begin{remark}[Kernel correction]
\label{rem:sec4_eta_kernel_correction}
If the self-paired sector is present, then the preceding reduced-\(\eta\) cancellation does not
hold in this simple form. Indeed, for a self-adjoint operator \(B\),
\begin{equation}
\eta(-B)=-\eta(B),
\qquad
h(-B)=h(B),
\end{equation}
where \(h(B)=\dim\ker B\). Therefore
\begin{equation}
\bar\eta(-B)
=
-\bar\eta(B)+h(B).
\end{equation}
Thus the reduced eta-invariant does not simply change sign in the presence of kernel. This is
why Proposition~\ref{prop:sec4_static_APS_index_zero} is stated under the invertibility assumption
\(-A_0\notin\mathcal K\).
\end{remark}

\begin{remark}
\label{rem:sec4_index_vs_reflection_trace} 
Thus, in the invertible fixed-holonomy line-twist model with flat twisting connection, the ordinary chiral APS index vanishes. When the self-paired sector is present, one must keep the finite-dimensional kernel correction described in Remark~\ref{rem:sec4_eta_kernel_correction}.
\end{remark}

\section{Reflection trace for the APS operator with reflection symmetry}
\label{sec:APSindex}

In this section, we study the part of the APS boundary problem that is visible to reflection symmetry.
Since the reflection element reverses orientation, the natural quantity is not a chiral equivariant
index, but rather the trace of the lifted reflection on the APS harmonic space
\begin{equation}\label{eq:sec5_APS_harmonic_space}
H^{\APS} = \ker(D^{\APS}).
\end{equation}
Here \(D^{\APS}\) denotes the APS Dirac operator, so \(H^{\APS}\) is the finite-dimensional space of
APS-harmonic spinors, that is, solutions of
\begin{equation}\label{eq:sec5_harmonic_equation}
D^{\APS}\psi = 0,
\qquad
\psi \in \Dom(D^{\APS}).
\end{equation}

Using the Fourier decomposition and the reflection pairing from the previous sections, we obtain
the decomposition
\begin{equation}\label{eq:sec5_kernel_decomposition}
\ker(D^{\APS})
=
\bigoplus_{\{k,k^\vee\}}
\Bigl(
(\ker(D^{\APS}) \cap \mathcal H_k)
\oplus
(\ker(D^{\APS}) \cap \mathcal H_{k^\vee})
\Bigr)
\oplus
(\ker(D^{\APS}) \cap \mathcal H_{-A_0}),
\end{equation}
where the direct sum runs over all non-self-paired reflection orbits
\begin{equation}\label{eq:sec5_orbit_decomposition}
\{k,k^\vee\},
\qquad
k^\vee = -k-2A_0,
\qquad
k \neq -A_0.
\end{equation}
The final summand \(\ker(D^{\APS}) \cap \mathcal H_{-A_0}\) is the contribution of the unique self-paired Fourier block, when it occurs.

Accordingly, the reflection-sensitive quantity studied in this section is 
\begin{equation}\label{eq:sec5_reflection_trace}
\chi_{\text{ref}}(r)
=
\Tr_{H^{\APS}} \mathcal U_r.
\end{equation}
The same Fourier-mode pairing as above will show that this reflection trace localizes to the
self-paired zero-mode sector.

\subsection{Fixed twist and kernel convention}

Assume throughout this section that
\begin{equation}\label{eq:sec5_fixed_twist}
A=A_0,
\qquad
2A_0\in\Z.
\end{equation}
Then the lifted reflection operator \(\mathcal U_r\) is globally defined. As in~\autoref{sec:rigorousSF}, its restriction to boundary data on \(Y_t\) for \(t\in\{0,T\}\), is denoted by
\begin{equation}\label{eq:sec5_boundary_reflection_operator}
\mathcal R_t:L^2(Y_t;S)\longrightarrow L^2(Y_t;S).
\end{equation}

In the self-paired sector, we use the completed APS convention from
Remark~\ref{rem:global_APS_convention}. When the self-paired mode \(k=-A_0\) occurs, the
corresponding boundary operator vanishes on that block. The auxiliary self-adjoint kernel
completion is assumed to be invariant under the boundary reflection \(\mathcal R_t\). Equivalently,
the completed APS projector satisfies
\begin{equation}\label{eq:sec5_kernel_invariance}
\mathcal R_t\,P_t^{\APS}\,\mathcal R_t^{-1}
=
P_t^{\APS}
\end{equation}
on the self-paired boundary block.

The APS operator used in this section is
\begin{equation}\label{eq:sec5_APS_operator}
D^{\APS}=D_{A_0}\big|_{\Dom(D^{\APS})},
\end{equation}
with domain
\begin{equation}\label{eq:sec5_APS_domain}
\Dom(D^{\APS})
=
\Bigl\{
\psi\in H^1(M;S):
P_0^{\APS}\bigl(\psi|_{Y_0}\bigr)=0,
\quad
P_T^{\APS}\bigl(\psi|_{Y_T}\bigr)=0
\Bigr\}.
\end{equation}

\subsection{Reflection preserves the APS boundary conditions}

\begin{proposition}[Reflection preserves APS boundary conditions]
\label{prop:sec5_reflection_preserves_APS_zero_mode}
Assume \eqref{eq:sec5_fixed_twist} and \eqref{eq:sec5_kernel_invariance}. Then
\begin{equation}\label{eq:sec5_domain_preservation}
\mathcal U_r\bigl(\Dom(D^{\APS})\bigr)=\Dom(D^{\APS}),
\end{equation}
and
\begin{equation}\label{eq:sec5_operator_commutation}
\mathcal U_r\,D^{\APS}=D^{\APS}\,\mathcal U_r.
\end{equation}
\end{proposition}

\begin{proof}
By Lemma~\ref{lem:reflection_intertwining},
\begin{equation}\label{eq:sec5_intertwining_repeat}
\mathcal U_r\,D\,\mathcal U_r^{-1}=D
\end{equation}
on smooth sections, so it remains only to check preservation of the APS boundary condition.

For each nonzero mode \(m=k+A_0\neq 0\), Proposition~\ref{prop:blockpair} gives the paired label, \(k^\vee=-k-2A_0\), \(k^\vee+A_0=-(k+A_0)=-m\).
Using the boundary formulas from \autoref{sec:setup},
\begin{subequations}
\begin{equation}\label{eq:sec5_boundary_formulas_repeat}
B_0^{(k)}=\frac{m}{f(0)}\sigma_3,
\qquad
B_T^{(k)}=-\frac{m}{f(T)}\sigma_3,
\end{equation}
and
\begin{equation}\label{eq:sec5_boundary_formulas_paired}
B_0^{(k^\vee)}=-\frac{m}{f(0)}\sigma_3,
\qquad
B_T^{(k^\vee)}=\frac{m}{f(T)}\sigma_3,
\end{equation}
\end{subequations} 
Let \(P_{t,k}^{\APS}\) denote the restriction of the global APS projector \(P_t^{\APS}\) to the \(k\)-th Fourier boundary block. 
By Lemma~\ref{lem:sec4_reflection_conjugates_boundary_blocks} and functional calculus, we obtain
\begin{equation} \label{eq:sec5_projector_transport_nonzero}
\mathcal R_t\,P_{t,k}^{\APS}\,\mathcal R_t^{-1}
=
P_{t,k^\vee}^{\APS}
\qquad (m\neq 0).
\end{equation}

Thus, the APS condition is preserved on every nonzero paired block.

It remains to consider the self-paired mode \(k=-A_0\), when it occurs. On that block,
\begin{equation}\label{eq:sec5_zero_mode_boundary_operator}
B_t^{(-A_0)}=0.
\end{equation}
The APS condition on this block is the auxiliary self-adjoint kernel completion fixed in
Remark~\ref{rem:global_APS_convention}. By the reflection-invariance assumption
\eqref{eq:sec5_kernel_invariance}, this completed kernel condition is preserved by
\(\mathcal R_t\). Hence the APS boundary condition is preserved on the zero-mode as well.
Hence \eqref{eq:sec5_domain_preservation} holds. Combining this with \eqref{eq:sec5_intertwining_repeat} yields \eqref{eq:sec5_operator_commutation}.
\end{proof}

\subsection{Reflection trace on the APS harmonic space}

Because the APS boundary problem is elliptic, its kernel is finite-dimensional. By Proposition~\ref{prop:sec5_reflection_preserves_APS_zero_mode}, the harmonic space defined in
\begin{equation}\label{eq:sec5_harmonic_kernel}
H^{\APS} = \ker D^{\APS}
\end{equation}
is preserved by \(\mathcal U_r\). We recall from \eqref{eq:sec5_reflection_trace} that the reflection trace is defined by 
\begin{equation}\label{eq:sec5_reflection_trace_}
\chi_{\text{ref}}(r)
= \Tr_{\ker D^{\APS}} \mathcal U_r.
\end{equation}

\begin{remark}
\label{rem:sec5_no_equivariant_index}
The quantity used here is simply the trace of the reflection operator on the full APS harmonic space. No graded or chiral equivariant-index formalism is needed for the localization statement below.
\end{remark}

\subsection{Localization of the reflection trace to the zero-mode sector}

\begin{theorem}[Paired-mode cancellation for the reflection trace]
\label{thm:sec5_reflection_trace_localization}
Assume \eqref{eq:sec5_fixed_twist} and \eqref{eq:sec5_kernel_invariance}. For every paired orbit \(\{k,k^\vee\}\) with \(k\neq k^\vee\), the contribution of that orbit to the reflection trace \eqref{eq:sec5_reflection_trace} vanishes. Hence
\begin{equation}\label{eq:sec5_zero_mode_localization}
\chi_{\text{ref}}(r)
= \Tr_{\ker D^{\APS}\cap \mathcal H_{-A_0}} \mathcal U_r.
\end{equation}
In particular, every contribution to nonzero reflection is fully supported in the self-paired zero-mode sector \(\mathcal H_{-A_0}\).
\end{theorem}

\begin{proof}
Fix a paired orbit \(\{k,k^\vee\}\) with \(k\neq k^\vee\), and set
\begin{equation}\label{eq:sec5_paired_harmonic_space}
K_{k,k^\vee}
=
(\ker D^{\APS}\cap \mathcal H_k)
\oplus
(\ker D^{\APS}\cap \mathcal H_{k^\vee}).
\end{equation}
Because \(D^{\APS}\) respects the Fourier decomposition and commutes with \(\mathcal U_r\), the space \(K_{k,k^\vee}\) is \(\mathcal U_r\)-invariant.

Choose orthonormal bases of \(\ker D^{\APS}\cap\mathcal H_k\) and \(\ker D^{\APS}\cap\mathcal H_{k^\vee}\). Since \(\mathcal U_r\) exchanges the two summands, the matrix of \(\mathcal U_r|_{K_{k,k^\vee}}\) has the block form
\begin{equation}\label{eq:sec5_offdiagonal_matrix}
\mathcal U_r\big|_{K_{k,k^\vee}}
=
\begin{pmatrix}
0 & U_{12}\\
U_{21} & 0
\end{pmatrix}.
\end{equation}
In particular,
\begin{equation}\label{eq:sec5_paired_trace_zero}
\Tr_{K_{k,k^\vee}} \mathcal U_r= 0.
\end{equation}
Summing over all non-self-paired orbits shows that every paired orbit contributes zero to the total reflection trace.

The only remaining Fourier block is the self-paired sector, which can occur only for the mode \(k=-A_0\), equivalently \(m=0\). Therefore
\begin{equation}
\chi_{\text{ref}}(r)
= \Tr_{\ker D^{\APS}\cap \mathcal H_{-A_0}} \mathcal U_r,
\end{equation}
which is exactly \eqref{eq:sec5_zero_mode_localization}.
\end{proof}

\begin{corollary}[Vanishing away from the self-paired sector]
\label{cor:sec5_trace_vanishes_without_zero_mode}
If \(-A_0\notin\mathcal K\), then the self-paired sector is absent and \(\chi_{\text{ref}}(r)=0\).
\end{corollary}

\begin{proof}
If \(-A_0\notin\mathcal K\), there is no self-paired sector. Theorem~\ref{thm:sec5_reflection_trace_localization} then shows that all contributions come from paired orbits, and each such contribution vanishes.
\end{proof}

\section{Pointwise \texorpdfstring{$O(2)$}{O(2)}-equivariant spectral flow at fixed holonomy}
\label{sec:O2_sf}

In the original line-twist model, a genuinely moving \(O(2)\)-equivariant family can occur only if the holonomy is fixed and the family is produced by an \(O(2)\)-equivariant bulk perturbation. However, the line-twist reflection symmetry from \autoref{sec:O2lift} is described there by a twisted rotation-reflection relation. To place the current family inside the standard \(RO(O(2))\)-valued equivariant spectral-flow framework, we restrict to the trivial gauge fixed-holonomy class in this section,
\begin{equation}
[A_0]=0,
\qquad\text{equivalently}\qquad
A_0\in\Z,
\end{equation}
and we work in the gauge \(A_0=0\). In that gauge, the APS family becomes genuinely pointwise \(O(2)\)-equivariant. We use \(RO(O(2))\) to denote the real representation ring of \(O(2)\). We study the resulting representation-valued block decomposition and the induced evenness statement for the ordinary spectral flow away from the self-paired sector.

\begin{lemma}[Gauge-triviality is necessary for genuine pointwise $O(2)$-equivariance in the line-twist model]
\label{lem:gauge_triviality_necessary}
Consider the fixed-holonomy line-twist Dirac operator with parameter \(A_0\).
Assume that the rotation action \(\rho(\varphi)\) together with the lifted reflection
\(\mathcal U_r\) defines an honest pointwise \(O(2)\)-action on the twisted Hilbert space in the
usual semidirect-product sense
\begin{equation}
\mathcal U_r\,\rho(\varphi)\,\mathcal U_r^{-1}=\rho(-\varphi),
\qquad
\varphi\in \mathbb R/2\pi\mathbb Z.
\end{equation}
Then
\begin{equation}
[A_0]=0 \text{ in } \mathbb R/\mathbb Z
\qquad\Longleftrightarrow\qquad
A_0\in \mathbb Z.
\end{equation}
Hence, after an integer gauge shift, we get \(A_0=0\).
\end{lemma}

\begin{proof}
By Proposition~\ref{prop:semidirect_relation}, whenever the reflection lift exists, we have \(\mathcal{U}_r\,\rho(\varphi)\,\mathcal{U}_r^{-1} = e^{-2iA_0\varphi}\rho(-\varphi)\).
Thus, if the rotation action together with the lifted reflection defines an honest
pointwise \(O(2)\)-action in the usual semidirect-product sense, then the scalar factor
must be trivial for all \(\varphi\), so \(e^{-2iA_0\varphi}=1\) for all \(\varphi\in \mathbb R/2\pi\mathbb Z\).
Equivalently, \([A_0]=0 \in \mathbb R/\mathbb Z\).
Since \(A_0\) is a real representative of its class modulo \(\mathbb Z\), this is equivalent to \(A_0\in\mathbb Z\).
After an integer gauge transformation, we may therefore choose the representative, \(A_0=0\).
In that gauge, the scalar factor in Proposition~\ref{prop:semidirect_relation} is identically \(1\),
and the usual semidirect relation holds.
\end{proof}

\subsection{A pointwise \texorpdfstring{$O(2)$}{O(2)}-equivariant bulk perturbation family} \label{subsec:O2_sf_family}

Fix a constant holonomy parameter \([A_0]=0\), equivalently \(A_0\in\Z\).
Since the line twist is gauge-trivial in this case, we choose the gauge \(A_0=0\).
The corresponding Dirac operator is 
\begin{equation}\label{eq:DA0_section6}
D_{A_0}
=
i\sigma_1\Bigl(\partial_t+\frac{f'(t)}{2f(t)}\Bigr)
+
i\sigma_2\,\frac{1}{f(t)}(\partial_\theta+iA_0)
\end{equation}
and let $D_{A_0}^{\APS}$ denote the same operator equipped with the APS boundary conditions from \autoref{sec:rigorousSF}.

We choose a smooth real-valued cutoff function
\begin{equation}\label{eq:chi_section6}
q\in C^\infty([0,T]),
\end{equation}
supported in the interior of the cylinder and vanishing in a neighborhood of
both boundary components. We also fix a real coupling parameter%
\footnote{The parameter \(\mu\) is included mainly for the physics-motivated interpretation of the perturbation as a coupling term. Mathematically, since the function \(q\) is a dimensionless smooth function, the overall scale can be absorbed into \(q\) and one may set \(\mu=1\) whenever \(\mu\neq 0\). We keep \(\mu\) explicit in order to track the strength of the perturbation.}
\begin{equation}\label{eq:mu_section6}
\mu\in\R.
\end{equation}
We then define the bounded self-adjoint multiplication operator
\begin{equation}\label{eq:V_section6}
V=\mu\,q(t)\,\Id_{\C^2},
\end{equation}
which may be viewed as a scalar bulk potential, or mass-type perturbation, supported away from the boundary. Finally, we consider the resulting one-parameter family
\begin{equation}\label{eq:Ds_section6}
D_s
=
D_{A_0}+sV,
\qquad
s\in[0,1].
\end{equation}

Since \(q(t)=0\) in a neighborhood of \(t=0\) and \(t=T\), the perturbation $V$ vanishes on a collar of the boundary. Hence \(D_s=D_{A_0}\) near \(Y_0\sqcup Y_T\), so the induced tangential boundary operators and, therefore, the APS boundary projections are unchanged.
Hence, the APS domain is independent of \(s\):
\begin{equation}\label{eq:APSdomain_section6}
\Dom(D_s^{\APS})
=
\Dom(D_{A_0}^{\APS})
\qquad
\text{for all }s\in[0,1].
\end{equation}
Accordingly, we write
\begin{equation}\label{eq:APS_operator_section6}
D_s^{\APS}
=
D_{A_0}^{\APS}+sV
\end{equation}
as an unbounded operator on \(L^2(M;S)\) with the common domain
\eqref{eq:APSdomain_section6}.

\begin{proposition}[Well-posed APS family]
\label{prop:wellposed_section6}
For every \(s\in[0,1]\), the operator \(D_s^{\APS}\) is self-adjoint and Fredholm on the
common domain \eqref{eq:APSdomain_section6}. Moreover, the path
\begin{equation}\label{eq:path_selfadjoint_fredholm_section6}
[0,1]\ni s\longmapsto D_s^{\APS}
\end{equation}
is norm-resolvent continuous.
\end{proposition}

\begin{proof}
The operator \(D_{A_0}^{\APS}\) is self-adjoint Fredholm by the APS theory for Dirac-type
operators on compact manifolds with boundary \cite{APS1,BoossWojciechowski1993}. Since \(V\) is bounded and self-adjoint in
\(L^2(M;S)\), each \(D_s^{\APS}\) is a bounded self-adjoint perturbation of
\(D_{A_0}^{\APS}\) in the same domain, hence is self-adjoint and Fredholm \cite{Phillips1996,Waterstraat2017}.

Finally,
\begin{equation}\label{eq:linear_dependence_section6}
D_s^{\APS}-D_{s'}^{\APS}=(s-s')V,
\end{equation}
so the resolvent identity implies norm-resolvent continuity.
\end{proof}

\begin{lemma}[Pointwise \texorpdfstring{$O(2)$}{O(2)}-equivariance]
\label{lem:O2_equiv_family}
Assume $A_0 \in \Z$, then the family
\(\{D_s^{\APS}\}_{s\in[0,1]}\) is pointwise \(O(2)\)-equivariant, that is, each
\(D_s^{\APS}\) commutes with the \(O(2)\)-action on \(L^2(M;S)\).
\end{lemma}

\begin{proof}
By the standing assumption \([A_0]=0\), the fixed holonomy class is gauge-trivial, and we work in the gauge \(A_0=0\). In this gauge, the line twist is trivial, so the rotation action \(\rho(\varphi)\) and the lifted reflection operator \(\mathcal U_r\) satisfy the ordinary \(O(2)\) group law on \(L^2(M;S)\). Hence the operator \(D_{A_0}\) is genuinely \(O(2)\)-equivariant.

The perturbation \(V=\mu\,q(t)\Id_{\C^2}\) depends only on \(t\), is scalar on the spinor fiber, and therefore commutes both with the circle action and with the lifted reflection \(\mathcal U_r\). Since the APS domain is fixed and reflection-compatible, the APS operator is \(O(2)\)-equivariant as well.
\end{proof}

For the remainder of this section, we use the following facts in the gauge-trivial fixed-holonomy case \(A_0=0\): the APS domain is independent of \(s\), each \(D_s^{\APS}\) is self-adjoint Fredholm on that common domain, the path \(s\mapsto D_s^{\APS}\) is norm-resolvent continuous, and the family is pointwise \(O(2)\)-equivariant. These follow from Proposition~\ref{prop:wellposed_section6} and Lemma~\ref{lem:O2_equiv_family}. Whenever signed crossing numbers are mentioned below, we additionally assume isolated regular crossings on the relevant restricted block paths; this extra assumption is needed only for local crossing-number interpretations.

\subsection{Mode decomposition and reflection pairing}
\label{subsec:O2_sf_modes}

As in \autoref{sec:setup}, the family decomposes into Fourier modes indexed by \(k \in \mathcal K\). 
Let \(\mathcal H_k \subset L^2(M;S)\) denote the corresponding Fourier subspace.
Then the restriction of \(D_s\) to the \(k\)-th mode is
\begin{equation}\label{eq:mode_operator_section6}
D_{s,k}
=
D_{A_0,k}
+
s\,\mu\,q(t)\,\Id_{\mathbb{C}^2}
=
i
\begin{pmatrix}
0 & A_k^+ \\
A_k^- & 0
\end{pmatrix}
+
s\,\mu\,q(t)\,\Id_{\mathbb{C}^2}, \quad 
A_k^\pm
=
\partial_t + \frac{f'(t)}{2f(t)} \pm \frac{m_k}{f(t)}.
\end{equation}
We write \(D_{s,k}^{\APS}\) for the operator \(D_{s,k}\) equipped with the APS boundary condition induced from the common
APS domain \(\Dom(D_s^{\APS})\).


\begin{remark}[Orbit labels, representation labels, and representation-ring classes]
\label{rem:orbit_rep_notation}
In the equivariant block decomposition we distinguish three levels of notation.

First, \(k\in\mathcal K_+\) denotes a chosen positive Fourier representative.

Second, \([k]\) denotes the corresponding non-self-paired reflection orbit (or orbit block)
in the line-twist model.

Third, \(\rho_k\) denotes the real two-dimensional irreducible \(O(2)\)-representation
carried by that nonzero block, and \([\rho_k]\in RO(O(2))\) denotes its class in the real
representation ring.

Thus, orbit brackets \([k]\) label blocks, whereas representation-ring brackets
\([\rho_k]\) label classes in \(RO(O(2))\).
\end{remark}

\begin{lemma}[Paired mode families]
\label{lem:paired_mode_families}
For every \(s\in[0,1]\) and every \(k\in\mathcal K\), the APS mode operators
\(D_{s,k}^{\APS}\) and \(D_{s,k^\vee}^{\APS}\) are unitarily equivalent.
\end{lemma}

\begin{proof}
By Lemma~\ref{lem:O2_equiv_family}, the full operator \(D_s^{\APS}\) commutes with the
lifted reflection \(\mathcal U_r\). By Proposition~\ref{prop:blockpair},
\(\mathcal U_r\) maps the \(k\)-th Fourier sector onto the \(k^\vee\)-th sector. Since the
perturbation \(V\) is scalar and depends only on \(t\), it is preserved under this
identification. Hence, the restriction of \(\mathcal U_r\) yields a unitary operator
\begin{equation}\label{eq:modewise_unitary_section6}
\mathcal U_{r,k}:\mathcal H_k\longrightarrow \mathcal H_{k^\vee}
\end{equation}
such that
\begin{equation}\label{eq:modewise_intertwining_section6}
\mathcal U_{r,k}\,D_{s,k}^{\APS}\,\mathcal U_{r,k}^{-1}
=
D_{s,k^\vee}^{\APS}.
\end{equation}
\end{proof}

For each non-self-paired orbit \([k]=\{k,k^\vee\}\), \(k\neq k^\vee\), set \(\mathcal H_{[k]}=\mathcal H_k\oplus \mathcal H_{k^\vee}\),
and
\begin{equation}\label{eq:orbit_block_operator_section6}
D^{\text{APS}}_{s,[k]}=D^{\text{APS}}_{s,k}\oplus D^{\text{APS}}_{s,k^\vee}.
\end{equation}
The non-self-paired orbit block $\mathcal H_{[k]}$ is an infinite-dimensional
$O(2)$-invariant Hilbert block built from the real two-dimensional angular
irreducible type $\rho_k$; equivalently, any regular crossing eigenspace in
this block carries the representation $\rho_k$.

\begin{proposition}[Orbit blocks, paired spectral flow, and irreducible type]
\label{prop:orbit_block_irreducible_section6}
Fix a non-self-paired reflection orbit \([k]=\{k,k^\vee\}\), \(k\neq k^\vee\).
Then the following holds.

\begin{enumerate}[label=\textup{(\roman*)}]
\item The subspace \(\mathcal H_{[k]}=\mathcal H_k\oplus\mathcal H_{k^\vee}\) is
\(O(2)\)-invariant, and the restricted path \(\{D_{s,[k]}^{\APS}\}_{s\in[0,1]}\) is pointwise \(O(2)\)-equivariant.

\item The non-self-paired orbit block $\mathcal H_{[k]}$ is an infinite-dimensional
$O(2)$-invariant Hilbert block built from the real two-dimensional angular
irreducible type $\rho_k$. Equivalently, any regular crossing eigenspace in
this block carries the representation $\rho_k$. We denote its class by \([\rho_{k}]\in RO(O(2))\).

\item The paired mode paths have the same ordinary spectral flow, \(\operatorname{sf}\!\bigl(D_{s,k}^{\APS}\bigr)
=
\operatorname{sf}\!\bigl(D_{s,k^\vee}^{\APS}\bigr)\).
Hence, if we set
\begin{equation}\label{eq:N_orbit_redefined_section6}
N_{[k]}
=
\operatorname{sf}\!\bigl(D_{s,k}^{\APS}\bigr)
=
\operatorname{sf}\!\bigl(D_{s,k^\vee}^{\APS}\bigr),
\end{equation}
then
\begin{equation}\label{eq:orbit_block_sf_twice_section6}
\operatorname{sf}\!\bigl(D_{s,[k]}^{\APS}\bigr)
=
2\,N_{[k]}.
\end{equation}
\end{enumerate}
\end{proposition}

\begin{proof}
Because the family is pointwise \(O(2)\)-equivariant, rotations preserve each Fourier sector
and the lifted reflection exchanges \(\mathcal H_k\) with \(\mathcal H_{k^\vee}\). Hence
\(\mathcal H_{[k]}\) is \(O(2)\)-invariant, and the restricted path on
\(\mathcal H_{[k]}\) is pointwise \(O(2)\)-equivariant.

Since \(k\neq k^\vee\), the orbit consists of two distinct angular frequencies which are
exchanged by reflection. The full space \(\mathcal H_{[k]}=\mathcal H_k\oplus\mathcal H_{k^\vee}\)
is therefore an infinite-dimensional \(O(2)\)-invariant Hilbert block, not literally a
two-dimensional representation. What is two-dimensional is the associated angular
\(O(2)\)-type, equivalently, the regular crossing eigenspace inside this block. We denote the
corresponding class in \(RO(O(2))\) by \([\rho_{k}]\).

By Lemma~\ref{lem:paired_mode_families}, the paths
\(\{D_{s,k}^{\APS}\}_{s\in[0,1]}\) and \(\{D_{s,k^\vee}^{\APS}\}_{s\in[0,1]}\) are
unitarily equivalent, so their ordinary spectral flows are equal. Since \(D_{s,[k]}^{\APS} = D_{s,k}^{\APS}\oplus D_{s,k^\vee}^{\APS}\), additivity of ordinary spectral flow over direct sums gives
\begin{equation}\label{eq:orbit_block_sf_additivity_section6}
\operatorname{sf}\!\bigl(D_{s,[k]}^{\APS}\bigr)
=
\operatorname{sf}\!\bigl(D_{s,k}^{\APS}\bigr)
+
\operatorname{sf}\!\bigl(D_{s,k^\vee}^{\APS}\bigr)
=
2\,N_{[k]}.
\end{equation}
\end{proof}

The only possible self-paired mode is \(k=-A_0\), equivalently \(m=0\). If \(k_{\text{self-paired}}\notin\mathcal K\), then no self-paired sector occurs.

\subsection{Equivariant spectral flow and the self-paired sector}
\label{subsec:O2_sf_formula}
\paragraph{Convention for equivariant spectral flow.}
Throughout this subsection, 
\begin{equation}
\operatorname{sf}_{O(2)}(D_s^{\APS}) \in RO(O(2))
\end{equation}
denotes the $RO(O(2))$-valued equivariant spectral flow of the pointwise $O(2)$-equivariant
self-adjoint Fredholm path $\{D_s^{\APS}\}_{s\in[0,1]}$.
For each non-self-paired orbit $[k]=\{k,k^\vee\}$, we write
\begin{equation}
N_{[k]}
= \operatorname{sf}(D_{s,k}^{\APS})
= \operatorname{sf}(D_{s,k^\vee}^{\APS}),
\end{equation}
as in Proposition~\ref{prop:orbit_block_irreducible_section6}.
Equivalently, $N_{[k]}$ is one half of the ordinary spectral flow of the full orbit block, \(\operatorname{sf}(D_{s,[k]}^{\APS}) = 2N_{[k]}\).
This is the normalization used in this section. If the self-paired mode \(k_{\text{self-paired}}=-A_0\) belongs to \(\mathcal K\), we denote by \(D_{s,\text{self-paired}}^{\APS}\)
the APS family on that self-paired sector and write
\begin{equation}\label{eq:wall_equivariant_sf_section6}
\operatorname{sf}_{O(2)}^{\text{self-paired}}(D_s^{\APS})
\end{equation}
for its equivariant spectral-flow contribution. If \(k_{\text{self-paired}}\notin\mathcal K\), or if the self-paired sector stays invertible for all \(s\), we set
\begin{equation}\label{eq:wall_term_zero_section6}
\operatorname{sf}_{O(2)}^{\text{self-paired}}(D_s^{\APS})=0.
\end{equation}

\begin{theorem}[Equivariant block decomposition and evenness off self-paired sector]
\label{thm:evenness_off_wall_section6}

Assume the gauge-trivial fixed-holonomy APS family
\(\{D_s^{\APS}\}_{s\in[0,1]}\) from \eqref{eq:Ds_section6}, with \([A_0]=0\), equivalently \(A_0\in\Z\), and with the gauge choice \(A_0=0\). Then
\begin{equation}\label{eq:ROO2_sf_section6}
\operatorname{sf}_{O(2)}(D_s^{\APS})
=
\sum_{\substack{[k]\\ k\neq k^\vee}}
N_{[k]}\,[\rho_{k}]
+
\operatorname{sf}_{O(2)}^{\text{self-paired}}(D_s^{\APS}),
\end{equation}
where the sum runs over all non-self-paired reflection orbits \([k]=\{k,k^\vee\}\). After
applying the dimension map
\begin{equation}\label{eq:dimension_map_section6}
\dim:RO(O(2))\longrightarrow \Z,
\end{equation}
one obtains
\begin{equation}\label{eq:integer_sf_section6}
\operatorname{sf}(D_s^{\APS})
=
2\sum_{\substack{[k]\\ k\neq k^\vee}} N_{[k]}
+
\dim\!\bigl(\operatorname{sf}_{O(2)}^{\text{self-paired}}(D_s^{\APS})\bigr).
\end{equation}
In particular, the theorem separates two levels of information:
the \(RO(O(2))\)-valued class records \emph{which} real \(O(2)\)-representation crosses,
whereas the ordinary spectral flow records only the total dimension after applying the dimension map. Further, every non-self-paired orbit contributes an even amount to the ordinary
spectral flow.
\end{theorem}

\begin{proof}
By Fourier decomposition and reflection pairing,
\begin{equation}\label{eq:orbit_decomposition_section6}
L^2(M;S)
=
\bigoplus_{\substack{[k]\\ k\neq k^\vee}}\mathcal H_{[k]}
\;\oplus\;
\mathcal H_{\text{self-paired}},
\end{equation}
where the last summand is present only when the self-paired mode belongs to \(\mathcal K\).
Because the family is pointwise \(O(2)\)-equivariant, each summand is \(O(2)\)-invariant, and the equivariant spectral flow is additive over orthogonal direct sums. Thus, the full equivariant spectral flow is the sum of the contributions of the non-self-paired orbit blocks together with the possible self-paired contribution.

For each non-self-paired orbit \([k]\), Proposition~\ref{prop:orbit_block_irreducible_section6}
shows that \(\mathcal H_{[k]}\) is an \(O(2)\)-invariant Hilbert block built from the
real two-dimensional angular representation \([\rho_{k}]\), and that the ordinary spectral flow of the
full orbit block is
\begin{equation}
\operatorname{sf}\!\bigl(D_{s,[k]}^{\APS}\bigr)=2N_{[k]}.
\end{equation}
Equivalently, the corresponding \(RO(O(2))\)-valued contribution is \(N_{[k]}[\rho_{k}]\). Summing these orbit contributions and adjoining the possible self-paired term gives \eqref{eq:ROO2_sf_section6}.

Applying the dimension map and using
\begin{equation}\label{eq:dimension_of_orbit_type_section6}
\dim([\rho_{k}])=2
\end{equation}
for every non-self-paired orbit yields
\eqref{eq:integer_sf_section6}. The final evenness statement is immediate.
\end{proof}

\begin{remark}[Scope of the self-paired term]
\label{rem:wall_term_scope_section6}
Theorem~\ref{thm:evenness_off_wall_section6} isolates the self-paired sector as the
only possible source of odd ordinary spectral flow. In the present line-twist model, the
statement is structural: it does not assert a closed formula for
\(\operatorname{sf}_{O(2)}^{\text{self-paired}}(D_s^{\APS})\), only that every non-self-paired orbit contributes through the class \([\rho_{k}]\) attached to the real two-dimensional
angular type, equivalent to the regular crossing eigenspace in that orbit block.
\end{remark}

\begin{corollary}[Mod-two consequence of the equivariant decomposition]
\label{cor:Z2_from_O2_sf}
Under the assumptions of Theorem~\ref{thm:evenness_off_wall_section6},
\begin{equation}\label{eq:Z2_from_RO_section6}
\operatorname{sf}(D_s^{\APS})
\equiv
\dim\!\bigl(\operatorname{sf}_{O(2)}^{\text{self-paired}}(D_s^{\APS})\bigr)
\pmod 2.
\end{equation}
In particular, if the self-paired sector is absent or remains invertible throughout the path, then
\begin{equation}\label{eq:integer_sf_even_section6}
\operatorname{sf}(D_s^{\APS})\equiv 0\pmod 2.
\end{equation}
\end{corollary}

\begin{proof}
This is immediate from \eqref{eq:integer_sf_section6}. In particular, if the self-paired sector is absent or remains invertible throughout the path,
then \(\operatorname{sf}(D_s^{\APS})\in 2\mathbb Z\).
\end{proof}

\begin{remark}[On regular crossings]
\label{rem:regular_crossings_section6}
The regular-crossing assumption is made only to keep the notation transparent. As usual, we can achieve regular crossings after an arbitrarily small perturbation within the class of gauge-trivial fixed-holonomy pointwise \(O(2)\)-equivariant paths with fixed invertible endpoints, and the resulting equivariant spectral-flow class is unchanged by homotopy. Thus, Theorem~\ref{thm:evenness_off_wall_section6} remains valid without this auxiliary assumption.

\end{remark}

\section{\texorpdfstring{$\mathbb Z_2$}{Z2} crossing parity for holonomy deformations}
\label{sec:Z2_flow_1}
In the original line-twist model, a nontrivial holonomy deformation cannot remain pointwise \(O(2)\)-equivariant along the full path. We can admit a path that allows the endpoints a reflection lift even though the intermediate operators fail to be pointwise \(O(2)\)-equivariant. We study the invariant along such a path: the parity of APS boundary-crossing events. This gives the appropriate sign-level replacement for the \(RO(O(2))\)-valued spectral flow of~ \autoref{sec:O2_sf}.

\begin{proposition}[\texorpdfstring{$\mathbb Z_2$}{Z2} crossing parity]
\label{prop:line_twist_bridge}
Consider the line-twist holonomy model with a continuous holonomy path
\begin{equation}\label{eq:continuous_holonomy_path_section8}
A:[0,1]\to\R.
\end{equation}
Assume that for every \(s\in[0,1]\), the corresponding APS operator remains in the
pointwise \(O(2)\)-equivariant line-twist setting of \autoref{sec:O2_sf}. Then \(A\) is
constant on \([0,1]\).

Consequently, any non-constant holonomy deformation in the line-twist model is not pointwise \(O(2)\)-equivariant along the full path. For such a deformation, the \(RO(O(2))\)-valued spectral flow of \autoref{sec:O2_sf} is generally unavailable along the full path, and the relevant invariant is the boundary crossing parity introduced below.
\end{proposition}

\begin{proof}
Assume that for every \(s\in[0,1]\), the corresponding line-twist APS operator remains in the
pointwise \(O(2)\)-equivariant setting of \autoref{sec:O2_sf}. By the analysis in
\autoref{sec:O2_sf}, this is exactly the gauge-trivial fixed-holonomy line-twist case. Hence,
for every \(s\in[0,1]\),
\begin{equation}
[A(s)]=0
\qquad\text{in }\R/\Z.
\end{equation}
Equivalently,
\begin{equation}
A(s)\in\Z
\qquad\text{for every }s\in[0,1].
\end{equation}
Since \(A:[0,1]\to\R\) is continuous and \(\Z\subset\R\) is discrete, the image \(A([0,1])\)
is connected and contained in a discrete set. Therefore \(A([0,1])\) consists of a single
point, so \(A\) is constant on \([0,1]\).

The conclusion now follows immediately. Any non-constant holonomy deformation in the original
line-twist model cannot remain in the pointwise \(O(2)\)-equivariant case along the full path.
As a result, the \(RO(O(2))\)-valued spectral flow from \autoref{sec:O2_sf} is generally
unavailable for the entire deformation, and the remaining invariant is the
\(\mathbb Z_2\)-valued boundary crossing parity we introduce below.
\end{proof}

\begin{remark}[Relation with \autoref{sec:O2_sf}]
\label{rem:bridge_evenness_to_Z2}
Proposition \ref{prop:line_twist_bridge} explains the dichotomy between the two moving-family cases treated in this paper.

In the gauge-trivial fixed-holonomy line-twist setting of \autoref{sec:O2_sf}, i.e. \([A_0]=0\) and the gauge choice \(A_0=0\), the ordinary spectral flow is even away from the self-paired sector, and the full invariant is the \(RO(O(2))\)-valued spectral flow; see Theorem \ref{thm:evenness_off_wall_section6} and Corollary \ref{cor:Z2_from_O2_sf}. However, a nontrivial holonomy deformation in the original line-twist model necessarily leaves that case. The representation-valued refinement is then lost, but the parity of the rank-one APS boundary crossing events survives. The present section isolates precisely this \(\mathbb Z_2\) datum.
\end{remark}

\subsection{\texorpdfstring{$\mathbb Z_2$}{Z2} crossing parity for holonomy deformations}
\label{subsec:Z2_flow_holonomy1}

\paragraph{Setup.}
Fix the Fourier lattice
\begin{equation}\label{eq:lattice_section8}
\mathcal K=\Z
\quad\text{(periodic spinors)},
\qquad
\mathcal K=\Z+\tfrac12
\quad\text{(anti-periodic spinors)}.
\end{equation}
Let $A:[0,1]\to\R $ be a smooth holonomy path, and for \(k\in\mathcal K\) we define
\begin{equation}\label{eq:m_sk_section8}
m(s,k)=k+A(s).
\end{equation}
The APS boundary data depend only on the sign of \(m(s,k)\) and can change only when
\begin{equation}
m(s,k)=0
\quad\Longleftrightarrow\quad
k+A(s)=0.
\end{equation}
For each \(s\in[0,1]\), let
\begin{equation}\label{eq:Ds_APS_section8}
D_s^{\APS}
\end{equation}
denote the APS operator corresponding to the holonomy value \(A(s)\).

\paragraph{Assumptions.}
Throughout this subsection, we assume:
\begin{equation}\label{eq:endpoint_invertibility_section8}
k+A(0)\neq 0,
\qquad
k+A(1)\neq 0
\qquad
\text{for all }k\in\mathcal K,
\end{equation}
and whenever
\begin{equation}\label{eq:wall_hit_section8}
k+A(s_*)=0,
\end{equation}
one has
\begin{equation}\label{eq:transversality_section8}
 A'(s_*)\neq 0.
\end{equation}
These hypotheses ensure that crossings are isolated and simple.

\paragraph{Convention on spectral flow in this subsection.}
Because the APS boundary condition changes at the crossing parameters \(k+A(s_*)=0\), the family \(s\mapsto D_s^{\APS}\) is in general only piecewise continuous as a family of self-adjoint Fredholm boundary problems. Therefore, throughout this subsection, the notation \(\operatorname{sf}(D_s^{\APS})\) refers to the \emph{crossing spectral flow} of the holonomy family, specifically the signed sum of the local simple crossings contributions established for this warped-cylinder APS model in \cite{KS}.
Away from the crossing parameters, this agrees with the usual spectral flow on each
continuity interval.

Moreover, since \(A([0,1])\subset\mathbb R\) is compact, only finitely many \(k\in\mathcal K\) can satisfy \(k+A(s)=0\) for some
\(s\in[0,1]\). Under the transversality assumption
\eqref{eq:transversality_section8}, all crossings are therefore finite in number and isolated.

To avoid confusion, this is not the ordinary spectral flow of a single continuous self-adjoint Fredholm path, but rather the boundary-crossing spectral flow defined by the modewise local crossing formula in \cite{KS}.

\begin{proposition}[Crossing formula for holonomy deformations]
\label{prop:wallcross}
Assume \eqref{eq:endpoint_invertibility_section8} and
\eqref{eq:transversality_section8}. Then the boundary crossing spectral flow of the APS holonomy family is
\begin{equation}\label{eq:integer_wallcross}
\operatorname{sf}\bigl(D_s^{\APS}\bigr)
=
\sum_{k\in\mathcal K}\ \sum_{s_*:\,k+A(s_*)=0}
\operatorname{sign}\bigl( A'(s_*)\bigr).
\end{equation}
Equivalently, each transverse crossing contributes \(\pm1\), with sign
determined by \(\operatorname{sign}(A'(s_*))\).
\end{proposition}

\begin{proof}
This is exactly the mode-by-mode boundary crossing formula established in
\cite{KS} for the warped-cylinder APS holonomy family. After Fourier
reduction, a crossing is confined to a single mode \(k\), in other words the one for
which \(m(s,k)=k+A(s)\) vanishes. Under the transversality hypothesis, that
crossing is simple, and its local contribution is
\(\operatorname{sign}(A'(s_*))\). Summing the local contributions over all
crossings and all Fourier modes yields \eqref{eq:integer_wallcross}.
\end{proof}

\begin{remark}[Modewise interpretation]
\label{rem:wallcross_summary}
Proposition~\ref{prop:wallcross} is the mode-by-mode boundary crossing formula: in the present
holonomy-deformation case, spectral flow is determined entirely by signed crossings of the self-paired sector \(k+A(s)=0\). When a full pointwise \(O(2)\)-equivariant structure is absent, this is
the correct replacement for an \(RO(O(2))\)-valued spectral-flow formula.
\end{remark}

\paragraph{Unsigned boundary crossing count and \texorpdfstring{$\mathbb Z_2$}{Z2} reduction.}
Define the total number of crossings by
\begin{equation}\label{eq:NA_def_section8}
N(A)
=
\#\bigl\{(s_*,k)\in(0,1)\times\mathcal K:\ k+A(s_*)=0\bigr\}.
\end{equation}

\begin{definition}[\texorpdfstring{$\mathbb Z_2$}{Z2} crossing parity]
\label{def:Z2_wall_parity}
The \(\mathbb Z_2\) boundary-parity invariant of the holonomy path \(A\) is
\begin{equation}\label{eq:SF_Z2_def_from_sf}
\operatorname{sf}_{\mathbb Z_2}(A)
=
\operatorname{sf}\bigl(D_s^{\APS}\bigr)\pmod 2
\in\mathbb Z_2.
\end{equation}
Equivalently,
\begin{equation}\label{eq:SF_Z2_def_from_N}
\operatorname{sf}_{\mathbb Z_2}(A)
=
N(A)\pmod 2.
\end{equation}
\end{definition}

\begin{theorem}[Parity formula]
\label{thm:Z2_parity_formula}
Under the standing assumptions,
\begin{equation}\label{eq:SF_Z2_wallcount}
\operatorname{sf}_{\mathbb Z_2}(A)
=
\sum_{k\in\mathcal K}
\#\bigl\{s_*\in(0,1):k+A(s_*)=0\bigr\}\pmod 2.
\end{equation}
Hence \(\operatorname{sf}_{\mathbb Z_2}(A)\) is simply the parity of the total number of boundary crossings.
\end{theorem}

\begin{proof}
Reduce the boundary crossing formula \eqref{eq:integer_wallcross} modulo \(2\). The signs
\(\operatorname{sign}(A'(s_*))\) disappear modulo \(2\), and only the total number of
crossings remains.
\end{proof}

\begin{proposition}[Local APS sign switch at a transverse boundary]
\label{prop:rank_one_jump_section8}
At a transverse crossing \(k+A(s_*)=0\), exactly one Fourier mode changes its
APS sign assignment, i.e. the mode \(k\) for which \(m(s_*,k)=0\).
Equivalently, in that mode, the boundary condition switches from the \(m>0\)
form \eqref{eq:APS_mpos} to the \(m<0\) form \eqref{eq:APS_mneg}, or vice
versa. Consequently, \(\operatorname{sf}_{\mathbb Z_2}(A)\) is the parity of
the number of simple APS sign-switch events along the path.
\end{proposition}

\begin{proof}
Away from the boundary zeros, the APS boundary condition depends only on the sign of
\(m(s,k)=k+A(s)\). At a transverse zero of \(m(s,k)\), exactly one Fourier
mode changes sign, and hence exactly one mode changes its APS boundary
assignment. Modulo \(2\), counting such simple sign-switch events is exactly
the boundary crossing count of Theorem~\ref{thm:Z2_parity_formula}.
\end{proof}

\begin{corollary}[Monotone case]
\label{cor:monotone_wall_parity}
If \(A\) is strictly monotone, then each equation \(k+A(s)=0\) has at most one solution, and
\begin{equation}\label{eq:SF_Z2_monotone}
\operatorname{sf}_{\mathbb Z_2}(A)
\equiv
\#\bigl(\mathcal K\cap(-A(1),-A(0))\bigr)
\pmod 2,
\end{equation}
with the obvious interval modification if \(A(1)<A(0)\).
\end{corollary}

\begin{proof}
Under strict monotonicity, a boundary zero can be crossed at most once for each \(k\in\mathcal K\). The
crossing parity is therefore obtained by counting which lattice points of \(\mathcal K\) lie
between the two endpoint values \(-A(0)\) and \(-A(1)\).
\end{proof}

\begin{remark}[Why endpoint data do not define an \texorpdfstring{$RO(O(2))$}{RO(O(2))}-valued flow]
\label{rem:endpoints_not_enough_section8}
If the path \(s\mapsto D_s^{\APS}\) fails to be pointwise \(O(2)\)-equivariant for intermediate \(s\), then the \(RO(O(2))\)-valued equivariant spectral flow is not defined for the full path, even if the endpoints admit a reflection lift. Nevertheless,
\(\operatorname{sf}_{\mathbb Z_2}(A)\) is always defined under the standing assumptions and provides a robust invariant of the deformation.
\end{remark}

\subsection{\texorpdfstring{$\mathbb Z_2$}{Z2} crossing parity and sign-level determinant language}
\label{subsec:pfaffian_Z2_tight}

We now record a sign-level determinant/Pfaffian interpretation of the same parity invariant. We do not introduce a new invariant, and nor construct an actual determinant-line orientation transport for the present complex APS family. Rather, we interpret the same \(\mathbb Z_2\)-class in a form closer to the determinant/Pfaffian sign language, in analogy with the real Fredholm setting.

\paragraph{Sign-level replacement in the non-\texorpdfstring{$O(2)$}{O(2)}-equivariant case.}
For the complex APS family considered here, we do not introduce a determinant-line orientation transport as a primary invariant. The remaining sign-level invariant is instead the mod-two crossing parity from Definition~\ref{def:Z2_wall_parity}. In the present
holonomy-deformation model, this \(\mathbb Z_2\)-valued quantity is the correct sign-level replacement for the representation-valued spectral flow available in the pointwise \(O(2)\)-equivariant case; compare the parity viewpoint for real Fredholm paths \cite{DollSchulzBaldesWaterstraat2019} and the determinant-line/orientation-transport formalism in the real setting \cite{NicolaescuOT}.

\paragraph{Warped cylinder}
In the present holonomy-deformation model, crossings occur exactly when
\begin{equation}\label{eq:wall_hit_repeat_section8}
k+A(s_*)=0.
\end{equation}
By the mode-by-mode boundary crossing analysis from
\cite{KS}, a transverse crossing affects exactly one Fourier
mode, and exactly one boundary eigenvalue crosses zero. Equivalently, the associated APS family undergoes a simple one-dimensional kernel event.

\begin{corollary}[Crossing parity as the remaining global sign]
\label{cor:endpoint_sign_wall_parity}
If the endpoints admit a reflection lift, but the path is not pointwise
\(O(2)\)-equivariant, then the corresponding sign-level invariant of the deformation is
\(\operatorname{sf}_{\mathbb Z_2}(A)\). Equivalently, it is the parity of the total number
of rank-one APS boundary crossing events.
\end{corollary}

\begin{proof}
By Proposition~\ref{prop:line_twist_bridge}, a non-constant holonomy path in the line-twist
model cannot remain pointwise \(O(2)\)-equivariant, so the representation-valued invariant
of \autoref{sec:O2_sf} is unavailable along the full path. By Theorem~\ref{thm:Z2_parity_formula} and Proposition~\ref{prop:rank_one_jump_section8}, the remaining sign-level data of the deformation is exactly the crossing parity
\(\operatorname{sf}_{\mathbb Z_2}(A)\), i.e., the parity of the simple APS
sign-switch events along the path.
\end{proof}

\subsubsection*{Example (periodic spin structure)}
Assume \(\mathcal K=\mathbb Z\), and consider the path
\begin{equation}\label{eq:OT_example_path_section8}
A(s)=\frac12+s,
\qquad
s\in[0,1].
\end{equation}
Then the endpoint invertibility condition
\eqref{eq:endpoint_invertibility_section8} holds, and there is exactly one boundary crossing, at \(s_*=\tfrac12\) for the mode \(k=-1\). Hence,
\begin{equation}\label{eq:OT_example_values_section8}
N(A)=1,
\qquad
\operatorname{sf}_{\mathbb Z_2}(A)=1.
\end{equation}
Thus, the path has odd crossing parity: exactly one simple APS sign-switch event
occurs along the deformation.

\begin{remark}[Domain-wall interpretation of the mod-two parity]
\label{rem:domain_wall_mod2_bridge}
The parity invariant \(\operatorname{sf}_{\mathbb Z_2}(A)\) admits a natural
conceptual interpretation in the light of the mod-two APS/domain-wall
correspondence. In the formulation of
\cite{FukayaFurutaMatsukiMatsuoOnogiYamaguchiYamashita2022}, the mod-two APS index of a real Dirac problem on a manifold with boundary is reformulated on a
closed manifold by replacing the boundary with a sign-changing mass term, i.e.
by a domain wall. In that picture, the mod-two index is read as a fermionic sign
associated with the wall system, and the bulk and edge contributions are
naturally separated.

The present warped-cylinder holonomy deformation is not itself rewritten here as
a domain-wall fermion problem, so we do \emph{not} claim a direct identification
theorem. Nevertheless, the role of
\(\operatorname{sf}_{\mathbb Z_2}(A)\) is fully consistent with that
perspective: when pointwise \(O(2)\)-equivariance is lost and the full
\(RO(O(2))\)-valued spectral flow of \autoref{sec:O2_sf} is unavailable, the deformation still carries a robust sign-level invariant, namely the parity
of the simple APS boundary crossing events. Thus, \(\operatorname{sf}_{\mathbb Z_2}(A)\) may
be viewed as the natural mod-two shadow of the more refined symmetry-resolved
spectral-flow data.
\end{remark}

\paragraph{Pfaffian and domain-wall language.}
For real skew-adjoint or Majorana-type formulations, one often packages the same
mod-two information as a Pfaffian sign: a simple eigenvalue crossing flips the
Pfaffian sign once. From this viewpoint, the parity class
\(\operatorname{sf}_{\mathbb Z_2}(A)\) may be regarded as the global fermionic
sign obstruction attached to the holonomy path. This interpretation is
compatible with the domain-wall reformulation of the mod-two APS index in
\cite{FukayaFurutaMatsukiMatsuoOnogiYamaguchiYamashita2022}, where the
boundary problem is replaced by a sign-changing mass term on a closed manifold.
We emphasize again that in the present paper, this is an interpretive analogy,
not an additional equivalence theorem.

\appendix
\section{Numerical illustration of the single-mode APS spectrum}
\label{Appendix:A}

We briefly record the numerical procedure used to generate the single-mode spectrum plot in the
reflection-compatible fixed-holonomy case. In the sample computation, we take
\begin{equation}
T=3,\qquad \alpha=1,\qquad A_0=\frac12,\qquad k=1,
\end{equation}
so that
\begin{equation}
m=k+A_0=\frac32>0.
\end{equation}
Thus, we are in a non-boundary-zero mode of the reflection-compatible sector
\begin{equation}
A=A_0,\qquad 2A_0\in\mathbb Z.
\end{equation}
We work with the warped factor
\begin{equation}
f(t)=e^t+\alpha e^{-t}=e^t+e^{-t},
\qquad t\in[0,T]=[0,3].
\end{equation}

Starting from the mode system
\begin{equation}
A^+v=-i\lambda u,\qquad A^-u=-i\lambda v,
\qquad
A^\pm=\partial_t+\frac{f'(t)}{2f(t)}\pm\frac{m}{f(t)},
\end{equation}
it is convenient to write
\begin{equation}
p(t)=\frac{f'(t)}{2f(t)},
\qquad
q(t)=\frac{m}{f(t)}.
\end{equation}
Then
\begin{equation}
A^+=\partial_t+p+q,
\qquad
A^-=\partial_t+p-q.
\end{equation}

Applying \(A^+\) to the equation \(A^-u=-i\lambda v\) gives
\begin{equation}
A^+A^-u=-\lambda^2 u,
\end{equation}
while applying \(A^-\) to \(A^+v=-i\lambda u\) gives
\begin{equation}
A^-A^+v=-\lambda^2 v.
\end{equation}
Hence \(u\) and \(v\) each satisfy a scalar second-order equation. To remove the first derivative
term, we set
\begin{equation}
u(t)=f(t)^{-1/2}U(t),
\qquad
v(t)=f(t)^{-1/2}V(t).
\end{equation}
A direct computation then yields the decoupled real equations
\begin{subequations}
\begin{equation}
U''+\bigl(\lambda^2-q'(t)-q(t)^2\bigr)U=0,
\end{equation}
and
\begin{equation}
V''+\bigl(\lambda^2+q'(t)-q(t)^2\bigr)V=0.
\end{equation}
\end{subequations}
Equivalently, since \(q(t)=m/f(t)\),
\begin{equation}
U''+\left(\lambda^2+\frac{m\,f'(t)}{f(t)^2}-\frac{m^2}{f(t)^2}\right)U=0,
\end{equation}
and
\begin{equation}
V''+\left(\lambda^2-\frac{m\,f'(t)}{f(t)^2}-\frac{m^2}{f(t)^2}\right)V=0.
\end{equation}

The APS boundary condition becomes a scalar boundary-value problem, and it is useful to record both sign cases separately.

\medskip
\noindent
\textbf{Case \(m>0\).}
In this case, the APS condition is
\begin{equation}
u(0)=0,
\qquad
v(T)=0.
\end{equation}
Using
\begin{equation}
A^-u=-i\lambda v
\end{equation}
at \(t=T\), the condition \(v(T)=0\) is equivalent to
\begin{equation}
A^-u(T)=0.
\end{equation}
Since
\begin{equation}
A^-\!\bigl(f^{-1/2}U\bigr)=f^{-1/2}(U'-qU),
\end{equation}
the scalar \(U\)-problem becomes
\begin{equation}
U(0)=0,
\qquad
U'(T)-q(T)U(T)=0.
\end{equation}
Thus, one may solve the \(U\)-equation with the normalized initial data
\begin{equation}
U(0)=0,
\qquad
U'(0)=1,
\end{equation}
and define the shooting function
\begin{equation}
S_U(\lambda)=U'(T;\lambda)-q(T)U(T;\lambda).
\end{equation}
Its zeros are exactly the APS eigenvalues of the chosen mode.

Equivalently, one may work with the \(V\)-equation. Using
\begin{equation}
A^+v=-i\lambda u
\end{equation}
at \(t=0\), the condition \(u(0)=0\) is equivalent to
\begin{equation}
A^+v(0)=0.
\end{equation}
Since
\begin{equation}
A^+\!\bigl(f^{-1/2}V\bigr)=f^{-1/2}(V'+qV),
\end{equation}
the scalar \(V\)-problem becomes
\begin{equation}
V'(0)+q(0)V(0)=0,
\qquad
V(T)=0.
\end{equation}
Thus, one may instead solve the \(V\)-equation with the normalized initial data
\begin{equation}
V(0)=1,
\qquad
V'(0)=-q(0),
\end{equation}
and define
\begin{equation}
S_V(\lambda)=V(T;\lambda).
\end{equation}
Again, its zeros are exactly the APS eigenvalues.

\medskip
\noindent
\textbf{Case \(m<0\).}
In this case, the APS condition is
\begin{equation}
v(0)=0,
\qquad
u(T)=0.
\end{equation}
Using
\begin{equation}
A^-u=-i\lambda v
\end{equation}
at \(t=0\), the condition \(v(0)=0\) is equivalent to
\begin{equation}
A^-u(0)=0.
\end{equation}
Hence the scalar \(U\)-problem becomes
\begin{equation}
U'(0)-q(0)U(0)=0,
\qquad
U(T)=0.
\end{equation}
A convenient normalization is
\begin{equation}
U(0)=1,
\qquad
U'(0)=q(0),
\end{equation}
and the corresponding shooting function is
\begin{equation}
S_U(\lambda)=U(T;\lambda).
\end{equation}

Equivalently, using
\begin{equation}
A^+v=-i\lambda u
\end{equation}
at \(t=T\), the condition \(u(T)=0\) is equivalent to
\begin{equation}
A^+v(T)=0,
\end{equation}
so the scalar \(V\)-problem becomes
\begin{equation}
V(0)=0,
\qquad
V'(T)+q(T)V(T)=0.
\end{equation}
Thus, one may solve the \(V\)-equation with the normalized initial data
\begin{equation}
V(0)=0,
\qquad
V'(0)=1,
\end{equation}
and define
\begin{equation}
S_V(\lambda)=V'(T;\lambda)+q(T)V(T;\lambda).
\end{equation}

In either sign case, the zeros of the corresponding scalar shooting function are precisely the APS
eigenvalues of the chosen Fourier block. For the numerical plot in the main text, we use the sample
values
\begin{equation}
T=3,\qquad \alpha=1,\qquad A_0=\frac12,\qquad k=1,
\qquad m=\frac32>0,
\end{equation}
so the displayed graph is the shooting function
\begin{equation}
S(\lambda)=S_U(\lambda)=U'(T;\lambda)-q(T)U(T;\lambda)
\end{equation}
associated with the \(U\)-equation. \autoref{fig:maslov_example3} is included only as a numerical illustration of the
single-block spectrum in the warped fixed-holonomy model; none of the proofs in the main text
depend on this computation.

\begin{figure}[t]
  \centering
  \includegraphics[width=.85\linewidth]{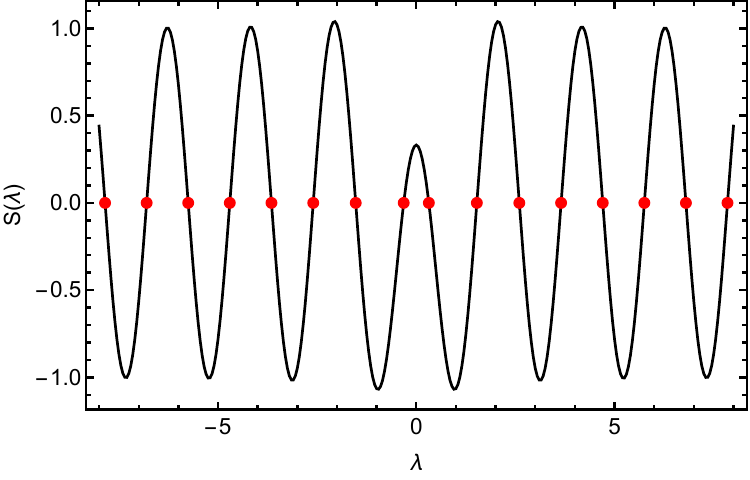}
\caption{Numerical illustration of the APS spectrum for the warped fixed-holonomy model with
\(T=3\), \(\alpha=1\), \(A_0=\frac12\), and \(k=1\), so that \(m=\frac32\). The warp factor is
\(f(t)=e^t+e^{-t}\). The plot shows the scalar function \(S(\lambda)\) obtained from the decoupled second-order equation. Its zeros are precisely the APS eigenvalues of the chosen non-self-paired mode, and the red circles indicate these eigenvalues.}
\label{fig:maslov_example3}
\end{figure}

\bibliographystyle{utphys}
\bibliography{ref}

\end{document}